\documentclass[11pt]{article}
\usepackage[letterpaper,margin=1in]{geometry}
\usepackage{xcolor}
\usepackage[colorlinks]{hyperref}
\definecolor{lightblue}{rgb}{0.5,0.5,1.0}
\definecolor{darkred}{rgb}{0.5,0,0}
\definecolor{darkgreen}{rgb}{0,0.5,0}
\definecolor{darkblue}{rgb}{0,0,0.5}

\hypersetup{colorlinks,linkcolor=darkred,filecolor=darkgreen,urlcolor=darkred,citecolor=darkblue}
\usepackage{pgfplotstable}
\usepackage{pgfplots}
\usepackage{caption}
\usepackage{subcaption}
\usepackage{tabularx}
\usepackage{verbatim}
\usepackage{amsthm}
\usepackage{amssymb}
\usepackage{graphicx}
\usepackage{mathtools}
\usepackage[linesnumbered,ruled,commentsnumbered,longend]{algorithm2e}
\usepackage{etoolbox}
\usetikzlibrary{graphs}
\usetikzlibrary{shapes.misc, positioning}
\usepackage[edges]{forest}
\usetikzlibrary{shapes,arrows,positioning,calc,decorations.pathmorphing}
\usetikzlibrary{positioning,arrows,arrows.meta}
\tikzset{
  xarrow/.style={arrows={-Latex[angle=100:2mm]}}
}
\usepackage{amsmath}
\newtheorem{lemma}{Lemma}
\newtheorem{corollary}[lemma]{Corollary}

\DeclareMathOperator{\Sym}{Sym}
\DeclareMathOperator{\id}{id}
\DeclareMathOperator{\Aut}{Aut}
\DeclareMathOperator{\Refx}{Ref}
\let\Sel\undefined
\DeclareMathOperator{\Sel}{Sel}
\DeclareMathOperator{\Inv}{Inv}
\DeclareMathOperator{\T}{T}
\DeclarePairedDelimiter{\ceil}{\lceil}{\rceil}

\definecolor{blue0}{RGB}{196, 196, 255}
\definecolor{blue1}{RGB}{140, 140, 255}
\definecolor{blue2}{RGB}{112, 112, 255}
\definecolor{blue3}{RGB}{0, 0, 255}

\usetikzlibrary{patterns}

\tikzset{
	hatch distance/.store in=\hatchdistance,
	hatch distance=10pt,
	hatch thickness/.store in=\hatchthickness,
	hatch thickness=2pt
}

\makeatletter
\pgfdeclarepatternformonly[\hatchdistance,\hatchthickness]{flexible hatch}
{\pgfqpoint{0pt}{0pt}}
{\pgfqpoint{\hatchdistance}{\hatchdistance}}
{\pgfpoint{\hatchdistance-1pt}{\hatchdistance-1pt}}%
{
	\pgfsetcolor{\tikz@pattern@color}
	\pgfsetlinewidth{\hatchthickness}
	\pgfpathmoveto{\pgfqpoint{0pt}{0pt}}
	\pgfpathlineto{\pgfqpoint{\hatchdistance}{\hatchdistance}}
	\pgfusepath{stroke}
}
\makeatother

\tikzset{pruned/.style={fill=gray,pattern=flexible hatch,hatch distance=3pt,hatch thickness=0.5pt,pattern color=lightgray}}

\tikzset{snake it/.style={decorate, decoration=snake, segment length=4mm, thick, draw=blue}}

\newcommand{\Traces}{\textsc{Traces}}
\newcommand{\nauty}{\textsc{nauty}}
\newcommand{\bliss}{\textsc{bliss}}
\newcommand{\saucy}{\textsc{saucy}}
\newcommand{\dejavu}{\textsc{dejavu}}
\newcommand{\conauto}{\textsc{conauto}}

\DeclareMathOperator{\Pru}{Prune}
\DeclareMathOperator{\Dev}{Dev}
\newcommand{\PruA}{\Pru_{\Inv}} 
\newcommand{\PruB}{\Pru_{\Aut}} 
\newcommand{\PruBW}{\Pru_{\Aut}'} 

\SetKwProg{Fn}{function}{}{end}\SetKwFunction{RandomAut}{Automorphisms}
\SetKwFunction{FirstLeaf}{FirstLeaf}
\SetKwFunction{RandomWalk}{RandomWalk}
\SetKwFunction{Sift}{Sift}
\SetKwFunction{RandomElement}{RandomElement}
\SetKwFunction{Refine}{Refine}
\SetKwFunction{RRef}{Ref}
\SetKwFunction{SSel}{Sel}
\SetKwFunction{Sift}{Sift}
\SetKwFunction{CertifyAutomorphism}{CertifyAutomorphism}

\makeatletter
\newenvironment{multicases}[1]
{\let\@ifnextchar\new@ifnextchar
	\left\lbrace%
	\array{@{}l*{#1}{@{\quad}l}@{}}}
{\endarray\right.\kern-\nulldelimiterspace}
\makeatother

\newcommand{\algorithmicdoinparallel}{\textbf{do in parallel}}
\makeatletter
\AtBeginEnvironment{algorithmic}{%
  \newcommand{\WhileP}[2][default]{\ALC@it\algorithmicwhile\ #2\ %
    \algorithmicdoinparallel\ALC@com{#1}\begin{ALC@for}}%
}

\newcommand\blfootnote[1]{%
  \begingroup
  \renewcommand\thefootnote{}\footnote{#1}%
  \addtocounter{footnote}{-1}%
  \endgroup
}

\tikzset{
  xarrow/.style={gray!66,arrows={-Latex[angle=90:2mm]}}
}

\usepackage{todonotes}

\begin{document}
\title{Parallel Computation of Combinatorial Symmetries}
\author{Markus Anders and Pascal Schweitzer\\ TU Darmstadt}
\maketitle

\begin{abstract}
	In practice symmetries of combinatorial structures are computed by transforming the structure into an annotated graph whose automorphisms correspond exactly to the desired symmetries. An automorphism solver is then employed to compute the automorphism group of the constructed graph. Such solvers have been developed for over 50 years, and highly efficient sequential, single core tools are available. However no competitive parallel tools are available for the task.

	We introduce a new parallel randomized algorithm that is based on a modification of the individualization-refinement paradigm used by sequential solvers. The use of randomization crucially enables parallelization. 

	We report extensive benchmark results that show that our solver is competitive to state-of-the-art solvers on a single thread, while scaling remarkably well with the use of more threads. This results in order-of-magnitude improvements on many graph classes over state-of-the-art solvers. In fact, our tool is the first parallel graph automorphism tool that outperforms current sequential tools.
\end{abstract}\blfootnote{The research leading to these results has received funding from the European Research Council (ERC) under the European Union's Horizon 2020 research and innovation programme (EngageS: grant agreement No.~{820148}).}
\section{Introduction} \label{sec:introduction}
Exploitation of symmetry has a dramatic impact on the efficiency of algorithms in various fields. 
This includes the fields of computer vision and computer graphics \cite{DBLP:journals/ftcgv/LiuHKG10}, automated reasoning~\cite{DBLP:journals/cacm/HeuleK17}, machine learning~\cite{DBLP:journals/jmlr/ShervashidzeSLMB11} and in particular convolutional neural networks~\cite{DBLP:conf/nips/GensD14}, mathematical programming~\cite{10.1007/978-1-4614-1927-3_9}, chemical databases~\cite{DBLP:conf/iccS/RuizG08},  SAT-solving~\cite{DBLP:conf/sat/KatebiSM10}, constraint programming~\cite{DBLP:reference/fai/GentPP06}, software verification~\cite{DBLP:conf/cav/ChuJ12}, model checking~\cite{DBLP:conf/forte/Godefroid99,DBLP:journals/csur/MillerDC06} and so on.

Before symmetries of a structure can be exploited, one first has to have algorithmic means to find the symmetries. For this, the structure is usually transformed into an annotated graph whose automorphisms correspond to the symmetries of the original structure. Then tools are employed that compute the graph's automorphism group. 

The current state-of-the-art implementations of solvers computing automorphism groups are \bliss{}~\cite{DBLP:conf/alenex/JunttilaK07}, \nauty{} and \Traces{}~\cite{McKay201494}, \conauto{} \cite{DBLP:conf/wea/Lopez-PresaCA13} as well as \saucy{}~\cite{Darga:2004:ESS:996566.996712}. 
All of the mentioned algorithms follow the individualization-refinement (IR) framework. 

These tools have become increasingly powerful through a multitude of techniques. 
Initially, each mentioned tool provided insightful new pruning ideas or new implementation tricks. However, many of these very diverse ideas from all the tools have transcended into all of the other tools by now. This  lead to a second generation comprised of improved versions of the tools.
The most recent and currently fastest solver is \Traces{} which excels at pruning the search space of difficult graphs. The tool has been meticulously engineered by Piperno over the past decade ever improving its performance. The tool is fastest on most graph classes, and on the few where it is not, it still performs competitively with the best of all the solvers. 

Recently the requirements of the application domains for the tools have changed. One major change stems from the different architecture of modern hardware. 
In fact, all of the aforementioned tools are  sequential, single-threaded applications: spreading the work load across multiple cores would align with the contemporary hardware trend of steadily increasing core counts. 
While there are some theoretical results, research on practical parallel isomorphism algorithms is quite limited.
In his thesis, Tener~\cite{Tener:2009:ADI:1925258} describes approaches to parallel isomorphism testing. However, he has subsequently not pursued this further in the last decade, and the described algorithm is based on algorithms of the first generation. The study~\cite{OpenMP} only performs comparisons against slow sequential algorithms.
Overall, parallel graph isomorphism testing has not witnessed any of the ideas that characterize the second generation of algorithms. Generally, fast isomorphism and automorphism algorithms have been persistently resistant to parallelization attempts in theory and practice.

The goal of this paper is to stimulate a third generation of isomorphism solvers by harnessing the power of a modern CPU for the computation of graph automorphisms. Parallelizing existing libraries is not a straightforward task, since the practical algorithms are based on the IR technique, which is a priori sequential.

\textbf{Contribution.} 
We introduce \dejavu{}, a novel randomized algorithm solving the automorphism group problem based on the IR paradigm.  
The tool 
is (1)~on a single thread competitive with currently fastest solvers available (sometimes even outperforming them), and
(2) on 8 threads outperforms the currently fastest solvers on most graph classes.
Using the de facto standard benchmark suite, we report extensive experimental results corroborating these two claims. 
The results also demonstrate the scalability of the tool as the number of threads is varied from 1 to 8.

\textbf{Underlying ideas and techniques.} 
The quintessence of our new algorithm, by which we achieve parallelizability, is to replace inherently sequential traversal strategies with randomized traversal:
\dejavu{} mainly performs repeated random root-to-leaf walks in the search tree (stemming from the IR-framework) in conjunction with a probabilistic abort criterion.
The main motivation is that computing multiple random root-to-leaf walks can be parallelized. 
To allow for an even split of the work load, various subroutines, most notably the so-called sifting algorithm which is used for the probabilistic abort criterion, also have to be parallelized efficiently.
However, to create a truly efficient tool, the algorithm has to be combined with further heuristics tailored to the new parallel, probabilistic setting.

\textbf{Randomization.} Our tool \dejavu{} is a randomized tool which, in principle, means that the output is not always correct. The idea of exploiting randomization originated from isomorphism testing algorithms~\cite{DBLP:conf/icalp/AndersS21, alenextoappear,DBLP:conf/alenex/KutzS07} which however neither compute automorphisms nor have any form of parallelization whatsoever. 
For these randomized approaches (including \dejavu{}), the user can set an error probability (e.g.~1\%) and the tool guarantees that for each input the probability of error is at most this number.
However, crucially our approach only has a 1-sided error. This means while some automorphisms may be missed when queried for the automorphism group of a graph, the solver guarantees that the output consists entirely of automorphisms of the graph. For applications exploiting symmetry this is the \emph{right kind of error}. This way, they may sometimes fail to exploit symmetries that were missed but this only slows down the running time. It does not lead to incorrect results for the application. 
\section{Preliminaries} \label{sec:preliminaries}
\subsection{Individualization-Refinement} \label{subsec:ir}
Following \cite{McKay201494} closely, we introduce IR algorithms. The summary is focused on the results necessary to describe automorphism computations and what is needed for our algorithms.

\textbf{Colored Graphs.} An undirected, finite graph $G = (V, E)$ consists of a set of vertices $V \subseteq \mathbb{N}$ and a set of edges $E \subseteq V^2$, where $E$ is symmetric. We always assume $V = \{1, \dots{}, n\}$. 

A coloring $\pi \colon V \to \{1, \dots{}, k\}$ is a surjective map, mapping vertices of a graph to \emph{cells} $1, \dots{}, k$. We call $\pi^{-1}(i) \subseteq V$ with $i \in \{1, \dots{}, k\}$ the \emph{$i$-th cell} of $\pi$, which is non-empty since~$\pi$ is surjective. With $|\pi| = k$ we denote the number of cells in a coloring. If $|\pi| = n$ holds, we call $\pi$ \emph{discrete}. Note that a discrete coloring also characterizes a permutation of $V$. 

A colored graph $(G, \pi)$ consists of a graph and a coloring. 
The symmetric group on~$\{1,\ldots,n\}$ is denoted~$\Sym(n)$. With $\Aut(G)$ we denote the automorphism group of a graph. An element $\varphi \in \Aut(G)$ is a permutation of vertices which maps the graph to itself, i.e., a bijective map $\varphi\colon V \to V$ where $G^\varphi := (\varphi(V), \varphi(E)) = (V, E) = G$ holds. For colored graphs we additionally require that the coloring is preserved, i.e., a vertex of a cell $c$ must be mapped to a vertex of cell $c$.
We thus define the colored automorphism group~$\Aut(G, \pi)$ as those permutations~$\varphi$ with $(G, \pi)^\varphi = (G^\varphi, \pi^\varphi) = (G, \pi)$. Note that in all of these definitions actual equality, e.g., equality of adjacency matrices and not isomorphism, is required. In the following, we only consider uncolored input graphs for the sake of simplicity. Let us remark, however, that we could use exactly the same machinery for colored graphs (see \cite{McKay201494}).

\textbf{Refinement.} In the following, we want to \emph{individualize} vertices and \emph{refine} colorings. Individualizing vertices in a coloring is a process that artificially forces the vertex into its own singleton cell. We use $\nu \in V^*$ to denote a sequence of vertices. In particular, we can record in such a sequence which vertices have been individualized.

A \emph{refinement} is a function $\Refx \colon G \times V^* \to \Pi$. Here~$\Pi$ is the set of colorings of~$V$, i.e., the set of ordered partitions of~$V$.
Given a graph $G$ and sequence of vertices $\nu$, it must satisfy the following properties: first, it is invariant under isomorphism, i.e., $\Refx(G^\varphi, \nu^\varphi) = \Refx(G, \nu)^\varphi$ holds for all $\varphi \in \Sym(n)$.
Secondly, it respects vertices in $\nu$ as being individualized, i.e., $\{v\}$ is a singleton cell in $\Refx(G, \nu)$ for all $v \in \nu$.

\textbf{Cell Selector.} If refinement classifies all vertices into different cells, determining automorphisms and isomorphisms for a graph is easy, after all, cells have to be preserved. Otherwise, individualization is used to artificially single out a vertex inside a non-singleton class. The task of a \emph{cell selector} is to isomorphism invariantly pick a non-singleton cell of the coloring. In the IR paradigm, all vertices of the selected cell will then be individualized one after the other using some form of backtracking. After the individualization, refinement is applied again and the process continues recursively.
Formally, a cell selector is a function $\Sel \colon G \times \Pi \to 2^V$ into the power set of~$V$ satisfying the following properties:
\begin{itemize}
	\item It is invariant under isomorphism, that is $\Sel(G^\varphi,\pi^\varphi) = \Sel(G,\pi)^\varphi$ holds for $\varphi \in \Sym(n)$.
	\item If $\pi$ is discrete then $\Sel(G,\pi) = \emptyset$.
	\item If $\pi$ is not discrete then  $|\Sel(G,\pi)| > 1$ and $\Sel(G,\pi)$ is a cell of $\pi$.
\end{itemize}

\textbf{Search Tree.} With the functions $\Refx$ and $\Sel$ at hand, we are now ready to define the \emph{search tree}. For a graph $G$ we use $\T_{(\Refx, \Sel)}(G)$ to denote the search tree of~$G$ with respect to refinement operator~$\Refx$ and cell selector~$\Sel$. The search tree is constructed as follows: each node of the search tree corresponds to a sequence of vertices of~$G$.
\begin{itemize}
	\item The root of $\T_{(\Refx, \Sel)}(G)$ is the empty sequence $\epsilon$.
	\item If $\nu$ is a node in $\T_{(\Refx, \Sel)}(G)$ and $C = \Sel(G,\Refx(G, \nu))$, then its children are $\{\nu.v \; | \; v \in C \}$, i.e., all extensions of~$\nu$ by one vertex~$v$ of~$C$.
\end{itemize}
With $\T_{(\Refx, \Sel)}(G, \nu)$ we denote the subtree of $\T_{(\Refx, \Sel)}(G)$ rooted in $\nu$. We omit indices $\Sel$ and $\Refx$ if they are apparent from context. Note that the leaves of a search tree correspond to discrete colorings of the graph, and therefore to permutations of $V$.

We recite the following crucial facts on isomorphism invariance of the search tree as given in \cite{McKay201494}, which follows from the isomorphism invariance of $\Sel$ and $\Refx$:
\begin{lemma} \label{lem:tree_invariant} For a graph $G$ and $\varphi \in \Sym(n)$ we have $\T(G)^\varphi = \T(G^\varphi)$. 
\end{lemma}
\begin{corollary} \label{lem:auto_tree_correspondence1} If $\nu$ is a node of $\T(G)$ and $\varphi \in \Aut(G)$, then $\nu^\varphi$ is a node of $\T(G)$ and $\T(G, \nu)^\varphi = \T(G, \nu^\varphi)$.
\end{corollary}
We have yet to mention how the search tree is used to find automorphisms of a graph: 
\begin{lemma} \label{lem:leaf_auto_correspondence2} If $\nu$ and $\nu'$ are leaves of $\T(G)$, then there exists an automorphism $\varphi \in \Aut(G)$ such that $\nu = \varphi(\nu')$ if and only if $\Refx(G, \nu')^{-1} \cdot \Refx(G, \nu)$ is an automorphism of $G$. 
\end{lemma}
We also say that~$\nu'$ is an \emph{occurrence} of~$\nu$ if there is some automorphism~$\varphi\in \Aut(G)$ for which~$\varphi(\nu')=\nu$.

\textbf{Pruning.} In the overall algorithm, we fix a single leaf $\tau$ and then search for automorphisms by comparing other leaves to it. We call this fixed leaf $\tau$ the \emph{target leaf}. Corollary~\ref{lem:auto_tree_correspondence1} and Lemma~\ref{lem:leaf_auto_correspondence2} show that this suffices to derive all automorphisms from the search tree. 

Unfortunately, however, the search tree itself can be exponentially large in the input~\cite{DBLP:conf/stoc/NeuenS18}. 
Therefore, we want to prune it as much as possible.  

Towards this goal, we define a \emph{node invariant} $\Inv \colon G \times V^* \to I$, which is a function mapping nodes of the tree to a totally ordered set $I$. We require some further properties:
\begin{itemize}
	\item The invariant must be isomorphism invariant, i.e., we require $\Inv(G, \nu_1) = \Inv(G^\varphi, \nu_1^\varphi)$ for all $\varphi \in \Sym(n)$.
	\item If $|\nu_1| = |\nu_2|$ and $\Inv(G, \nu_1) < \Inv(G, \nu_2)$, then for all leaves $\nu_1' \in \T(G, \nu_1)$ and $\nu_2' \in \T(G, \nu_2)$ we require $\Inv(G, \nu_1') < \Inv(G, \nu_2')$.
\end{itemize}
It follows that even if we remove all nodes of the tree whose invariant deviates from the corresponding node invariant on the same level on the path to the target leaf, we can still retrieve the entire automorphism group. 
This operation is called \emph{pruning using invariants}.
Formally, we define $\PruA(\tau', \nu')$ to denote the operation that removes the subtree of node $\nu'$ if $\Inv(G, \tau') \neq \Inv(G, \nu')$, where $|\tau'| = |\nu'|$ holds and $\tau'$ is the prefix of length~$|\nu'|$ of $\tau$. 

We now describe \emph{pruning using automorphisms}. 
Assume we already have $\varphi \neq \id$ of $\Aut(G)$ available. For nodes $\nu$ where $\nu^\varphi$ is \emph{not} a prefix of the target leaf, we define $\PruB(\nu, \nu^\varphi)$ to denote the operation which removes the subtree rooted at $\nu^\varphi$ from the search tree. 
Applying $\PruB$ can only cut away parts of the search tree which are generated by the already available automorphisms anyway \cite{McKay201494}.

\subsection{Schreier-Sims Fundamentals} \label{subsec:randomschreiersims}
The procedure to aggregate automorphisms of \dejavu{} works on similar principles as the random Schreier-Sims algorithm, which provides us with a data structure to dynamically manage permutation groups. 
To be more precise, our algorithm needs a way to determine whether a newly found automorphism $\varphi$ is in the group generated by the automorphisms that were found previously. 
The procedure we use for this is called \emph{sifting}.
We give a brief description following the lines of \cite{seress_2003}.

All groups we consider are permutation groups $\Gamma \leq \Sym(\Omega)$. For the domain we always set $\Omega = \{1, \dots{}, n\}$. By $\langle S \rangle$ we denote the group generated by the elements of $S$, i.e., all elements that can be written as a product of elements of $S$. If $\langle S \rangle = \Gamma$ holds, we call $S$ a \emph{generating set} of $\Gamma$. 

\SetKw{Break}{break}
\SetKwFor{For}{for (}{)}{}
\begin{algorithm}[t]
	\SetAlgoLined
	\SetAlgoNoEnd
	\caption[Sifting Procedure]{Sifting} \label{alg:sift}
	\Fn{\Sift{S, T, B, $\varphi$}}{
		\SetKwInOut{Input}{Input}
		\SetKwInOut{Output}{Output}
		\Input{generators $S$, transversal table $T$, base points $B$, element $\varphi$}
		\Output{whether $S$, $B$ and $T$ remained unchanged}
		\For{$i = 1;\ i \leq |B|;\ i = i + 1$}{
			$b_i$ := $\varphi(B_i)$\;
			$t$ := $(T_i)_{b_i}$\;
			\lIf{$t = \bot$}{\Break
			}
			$\varphi$ := $\varphi \cdot t^{-1}$\; \label{line:sift:multiply}
		}
		\If{$\varphi \neq \id$}{
			$S$ := $S \cup \{\varphi\}$\;
			$b_i$ := $\varphi(B_i)$\;
			$(T_i)_{b_i}$ := $\varphi$\;
			\Return{false}\;
		}
		\Return{true}\;
	}
\end{algorithm}
We need the notion of a \emph{pointwise stabilizer} of a permutation group $\Gamma \subseteq \Sym(\Omega)$. Let $\beta \in \Omega$ be a point, then 
$\Gamma_{(\beta)} := \{\varphi \in \Gamma \; | \; \varphi(\beta) = \beta \}$.
For a sequence of points $(\beta_1, \dots{}, \beta_m) \in \Omega^m$ we just recursively take the pointwise stabilizer of all elements: 
\[\Gamma_{(\beta_1, \dots{}, \beta_m)} := \begin{multicases}{2}\Gamma & \text{ if } m = 0 \\
(\Gamma_{(\beta_1, \dots{}, \beta_{m-1})})_{(\beta_m)} & \text{ otherwise. }  
\end{multicases}
\]

\noindent We call a sequence of points $B = (\beta_1, \dots{}, \beta_m) \in \Omega^m$ a \emph{base} relative to $\Gamma \leq  \Sym(\Omega)$ if $\Gamma_{B} = \{\id\}$. 
For a generating set $\langle S \rangle = \Gamma$ and a base $(\beta_1, \dots{}, \beta_m)$ we define $S_i = S \cap \Gamma_{(\beta_1, \dots{}, \beta_i)}$. We call $S$ \emph{strong} relative to $\Gamma$ and $(\beta_1, \dots{}, \beta_m)$ if $\langle S_i \rangle = \Gamma_{(\beta_1, \dots{}, \beta_i)}$ holds for all $i \in \{0, \dots{}, m\}$.

Given a subgroup $\Delta \leq \Gamma$, a \emph{transversal} of $\Delta$ in~$\Gamma$ is a subset~$T\subseteq \Gamma$ that satisfies $|T \cap gH| = 1$ for every coset $g\Delta$ of $\Delta$ in~$\Gamma$.  We construct a \emph{transversal table} for a given base $B$ and generating set $S$, which contains a transversal for each subgroup $\langle S_i \rangle$~in~$\langle S_{i-1} \rangle$. We refer with $T_i$ to the transversal of $S_i$. Careful inspection of the definition reveals that each $\langle S_i \rangle$ fixes the $i$-th base point of $B$, i.e., for all $\varphi \in \langle S_i\rangle$ it is true that $\varphi(\beta_i) = \beta_i$. If we want to know the cosets of $S_i$ in $\langle S_{i - 1} \rangle$, we need to find the possible images of~$\beta_i$ in $\langle S_{i - 1} \rangle$. Elements of $\langle S_{i - 1} \rangle$ under which~$\beta_i$ has the same image are in the same coset of~$\langle S_i\rangle$.
Thus, we can differentiate transversal elements $T_i$ according to the image of~$\beta_i$ under them.
We denote by $(T_i)_b$ the element in $T_i$ mapping $\beta_i \mapsto b$ if it exists. We set $(T_i)_b = \bot$ if such an element does not exist. The cosets correspond to the orbit of $\beta_i$ in $S_{i - 1}$. Given an element $\varphi \in S_{i - 1}$, we need to determine to where $\varphi$ maps $\beta_i$ in order to find the coset in which it is contained. The representative of the coset is that element~$t$ in the transversal $T_i$ which also maps $\beta_i$ to $\varphi(\beta_i)$. By forming the product $\varphi \cdot t^{-1} \in S_i$ we obtain an element that fixes $\beta_i$. 

Algorithm~\ref{alg:sift} describes a \emph{sifting} procedure, which can be used to test membership in a given permutation group whenever a strong generating set $S$ and corresponding base $B$ are available. Otherwise, if $S$ is not strong or $B$ not complete, the sifting procedure computes a non-trivial permutation. In the version of the algorithm described here, this permutation is added to the generating set to ensure that now the sifted element is covered. Possibly one needs to extend the base for this purpose. If an element \emph{sifts successfully}, i.e., the procedure returns \textit{true}, we know that it is contained in $\langle S \rangle$. 
On the other hand, if the sifting is unsuccessful then the element was not in the group or the generating set was not strong and has been extended towards ensuring it to become strong.

The algorithm repeatedly multiplies transversal elements to the initial element. The operations preserve the property of whether the initially given element is in the group. Each operation modifies the element so that it is contained the next respective pointwise stabilizer. 

We refer to base, transversal table and generating set together as a \emph{Schreier structure}. As elements are sifted, such a structure captures the progress made towards constructing the group. 
A crucial result we exploit is the following, related to Lemma~4.3.1 in~\cite{seress_2003}: 
\begin{lemma} \label{lem:group_abort_criterion2} Let $\Gamma$ be a group, $B$ a base, $S$ a set of permutations in~$\Gamma$ and $\varphi$ a uniformly distributed element in $\Gamma$. If $\langle S \rangle \neq \Gamma$, the probability that $\varphi$ does not successfully sift through the Schreier structure defined by $B$ and $S$ is at least $\frac{1}{2}$.
\end{lemma}
\noindent The previous results are also the foundation for the random Schreier-Sims method, which is used by all competitive solvers to detect possibilities to apply the pruning function $\PruB$. 

Let us also record that the individualized vertices in a leaf of the search tree (defined in Section~\ref{subsec:ir}) actually form a base of the respective automorphism group \cite{McKay201494}. 
\section{Parallel Computation of Automorphisms} \label{sec:randomir}
We first describe how to turn random walks on IR trees into a correct, probabilistic algorithm. 
Then, we discuss how to parallelize sifting as required by the algorithm.
Lastly, we augment the algorithm using breadth-first traversal into the underlying procedure of \dejavu{}. 

The motivation is that the three fundamental methods mentioned above parallelize efficiently as long as the IR tree is sufficiently large.

\subsection{Random Walks and Automorphisms} \label{subsec:solvingrandomwalks}
The first step of our algorithm is to compute a random walk to one of the leaves, the target leaf. 
The goal is then to find another occurrence target leaf through random walks, whereby automorphisms are found. A key observation is that by choosing uniform, random walks through the tree --- which we describe in Algorithm~\ref{alg:random_walk} --- we also get a uniform distribution of elements in the automorphism group. The algorithm applies the refinement to the input graph and then repeatedly chooses a uniform random vertex of the target cell chosen by the cell selector for individualization. Starting from the initial coloring, it then keeps individualizing and refining until the coloring becomes discrete. It returns the coloring and the sequence of individualized vertices.

Recall that we refer to a leaf $\tau'$ as an \emph{occurrence} of $\tau$ if~$\tau'$ can be mapped to $\tau$ using an element $\varphi \in \Aut(G)$ (i.e., $\varphi(\tau') = \tau$). In this situation we call $\varphi$ the \emph{corresponding automorphism} with regard to $\tau'$. Note that there is a unique occurrence of $\tau$ for every $\varphi \in \Aut(G)$:
\begin{lemma} \label{lem:leaf_auto_correspondence0} A leaf $\tau$ can be mapped to exactly $|\Aut(G)|$ leaves in $\T(G)$ using elements of the automorphism group $\Aut(G)$.
\end{lemma}
\begin{proof} Note that $\tau$ is a base of $\Aut(G)$. Now consider an element $\varphi \in \Aut(G)$. Clearly, $\tau^\varphi$ also corresponds to a leaf in the tree (Lemma~\ref{lem:auto_tree_correspondence1}) and $\tau^\varphi$ is a base as well. Now consider a different element $\varphi' \in \Aut(G)$, i.e., $\varphi' \neq \varphi$. Clearly, $\tau^\varphi \neq \tau^{\varphi'}$ holds since $\tau$ is a base. 
\end{proof}

 \begin{lemma}\label{lem:leaf_auto_correspondence0p5}
	As a random variable, the output of Algorithm~\ref{alg:random_walk}, which is a leaf in the search tree, is uniformly distributed within each equivalence class of leaves. 
		\end{lemma}
		\begin{proof} There is a unique occurrence of $\tau$ for every automorphism (Lemma~\ref{lem:leaf_auto_correspondence0}). Hence, it suffices to argue that the probability of finding each occurrence of~$\tau$ through a random walk in the tree is equal. Assume that we are in a node $\nu$ of the search tree and let $\nu_1, \dots, \nu_k$ be the children of $\nu$. Let $\nu_1', \dots{}, \nu_k'$ be the children that correspond to the subtrees of $\nu$ that do contain an occurrence of $\tau$. Since we are sampling an element uniformly from $\nu_1, \dots{}, \nu_k$  in Algorithm~\ref{alg:random_walk}, each of these subtrees has the same probability of being chosen. Therefore, it suffices to argue that the chance of finding an occurrence of $\tau$ in each of $\nu_1', \dots{}, \nu_k'$ is equal. Since they all contain an occurrence of $\tau$, they can all be mapped to each other using the corresponding automorphisms. But this immediately implies that all of these subtrees must be isomorphic (Lemma~\ref{lem:tree_invariant}), showing the claim. 
		\end{proof}	
The following lemma immediately follows.
\begin{lemma}\label{lem:leaf_auto_correspondence1}
Let~$\tau$ be a fixed leaf. Consider the distribution of outputs of Algorithm~\ref{alg:random_walk} under the condition that an occurrence of~$\tau$ is computed. For such a given output~$\tau'$ consider the automorphism~$\varphi$ with~$\varphi(l)=l'$ corresponding~$\tau'$. Then~$\varphi$ is uniformly distributed in~$\Aut(G)$.
\end{lemma}
So, Algorithm~\ref{alg:random_walk} provides us with a method to uniformly sample random automorphisms. We now need a method to collect these automorphisms and determine when we have found enough of them to generate the entire automorphism group.

\emph{Description of Algorithm~\ref{alg:random_auto}.} 
The algorithm repeatedly samples automorphisms from the automorphism group through random walks (using Algorithm~\ref{alg:random_walk}). Then, it uses a probabilistic test based on Lemma~\ref{lem:group_abort_criterion2} and Lemma~\ref{lem:leaf_auto_correspondence1} to determine termination. When a certain number~$d=\ceil{-\log_2(\frac{\varepsilon}{2})}$ of \emph{consecutively} sampled automorphisms turn out to be already covered by the previously found automorphisms (i.e., they sift successfully) the algorithm terminates. 
The initial value of~$d$ is linked to the guaranteed bound on the error probability~$\varepsilon$ that can be chosen by the user. To guarantee that the error bound is kept, when some but not~$d$ consecutively found automorphisms were discovered, the value of~$d$ is incremented. 

\begin{algorithm}[t]
	\SetAlgoLined
	\SetAlgoNoEnd
	\caption[Random Walk of the Search Tree]{Random Walk of the Search Tree} \label{alg:random_walk}
	\Fn{\RandomWalk{$G$}}{
		\SetKwInOut{Input}{Input}
		\SetKwInOut{Output}{Output}
		\Input{graph $G$}
		\Output{a random leaf of the search tree and the individualized vertices}
		$base$  := ()\;
		$col$  := \RRef{G, $[v \mapsto 1]$, base}\;
		$cell$    := \SSel{G,col}\;
		\While{$cell \neq \emptyset$}{
			$v$ := \RandomElement{cell}\;
			$base$ := $base.v$\tcp*{append~$v$ to $base$}
			$col$ := \RRef{G, col, base}\;
			$cell$   := \SSel{G, col}\;
		}
		\Return{(col, base)}\;
	}
\end{algorithm}

\begin{algorithm}[t]
	\SetKwFor{While}{while}{do in parallel}{end}
	\SetAlgoLined
	\SetAlgoNoEnd
	\caption[Randomized Automorphism Group]{Parallel Randomized Automorphisms} \label{alg:random_auto}
	\Fn{\RandomAut{$G$, $\epsilon$}}{
		\SetKwInOut{Input}{Input}
		\SetKwInOut{Output}{Output}
		\Input{graph $G$ and probability $\varepsilon$}
		\Output{a subset of~$\Aut(G)$ that generates~$\Aut(G)$ with probability at least~$1-\varepsilon$}
		$c$    := $0$, $d$    := $\ceil{-\log_2(\frac{\varepsilon}{2})}$, $S$    := $\emptyset$\;
		$(\tau, B)$    := \RandomWalk{G}\;
		initialize trivial tranversal table $T$ relative to $B$\;

		\While{$c \leq d$}{ \tcp{run multiple instances of the body of the loop in parallel}\label{line:main:loop}
			$(l', \_)$ := \RandomWalk{G}\;
			$\varphi$ := $l' \cdot \tau^{-1}$\;
			\If{$\varphi(G) = G$}{
				\lIf{$\neg$ \Sift{$S$, $T$, $B$, $\varphi$}}{$c$ := $c + 1$}\label{line:c:increment} 
				\Else{
					\lIf{$c > 0$}{$d$ := $d + 1$}
					$c$ := 0}
			}
		}
		\Return{$S$}\;
	}
\end{algorithm}

Finishing the execution therefore hinges on seeing already explored leaves as well as already generated automorphisms again (and hence the name ``\dejavu{}'').
Note that the correctness of the algorithm depends on the fact that we are probing automorphisms uniformly from the group. In Section~\ref{subsec:uniformpruning}, we introduce further techniques to prune the search tree. When we do so, we always make sure to do this in a manner that still enables us to probe uniformly after the pruning. Ensuring this suffices to retain a correct behavior of the algorithm.

We now argue correctness for Algorithm~\ref{alg:random_auto}.
\begin{lemma}\label{lem:error:prob:of:algo} Given a graph $G$ and probability $\varepsilon$, Algorithm~\ref{alg:random_auto} produces a generating set for the automorphism group of $G$ with probability at least $1 - \varepsilon$.
\end{lemma}
\begin{proof} First, observe that the discovered permutations are certified before being added to the group, which immediately ensures that all elements of the computed group are actual automorphisms. The algorithm can therefore only fail by not adding enough elements to the group.
	
Choosing random walks through the tree produces a uniform distribution of occurrences of the leaf $\tau$, which gives us a uniform distribution of elements in $\Aut(G)$ (Lemma~\ref{lem:leaf_auto_correspondence1}). This in turn enables us to use Lemma~\ref{lem:group_abort_criterion2} to argue correctness as follows.
	
We terminate the algorithm after we sifted uniform elements of $G$ successfully into the Schreier structure $d$ times in a row. As long as sifting fails and we add elements to the Schreier structure, we know that no error occurs and that we are not done. We view the computation as a sequence of \emph{tests} against the hypothesis that we are missing automorphisms. 
We define the beginning of a test to be right after sifting succeeds once (i.e., at the moment when $c$ is set to~$1$ in an execution of Line~\ref{line:c:increment}). The probability that the test fails (i.e., that we do not abort the test early and instead increment~$c$ for~$d$ times in a row) is bounded by $(\frac{1}{2})^d$ (Lemma~\ref{lem:group_abort_criterion2}). In order to ensure a total error bound of $\varepsilon$ for the algorithm, we require the sum of the probabilities of the tests to fail is at most~$\varepsilon$.
For this it suffices to have that the $i$-th test fails with probability at most $\frac{\varepsilon}{2^i}$. The probability that the entire computation fails is then surely at most $\varepsilon$ since
$\sum_{i = 1}^{\infty}  \frac{\varepsilon}{2^i} \leq \varepsilon.$

\noindent In order to satisfy this bound of $\frac{\varepsilon}{2^i}$, we increment $d$ after each successful test. Initially, for the first test, we set~$d_1 = \lceil -\log_2(\frac{\varepsilon}{2})\rceil $ which ensures that~$(\frac{1}{2})^{d} \leq \frac{\varepsilon}{2}$. Note the value~$d_i$ for variable~$d$ used during the~$i$-th test is then~$d_i = d_1 +i-1$, so~$(\frac{1}{2})^{d_i} < \frac{\varepsilon}{2^i}$, as desired.
\end{proof}
We should clarify that while the algorithm is based on some of the same principles as the isomorphism test of \cite{alenextoappear}, that isomorphism test neither has to consider uniformity of automorphism sampling (Lemma~\ref{lem:leaf_auto_correspondence1}), nor employ repeated testing, nor requires any form of sifting. 

Through the use of the randomized approach, a simple opportunity for parallelization arises by running the body of the while-loop in Line~\ref{line:main:loop} on multiple threads. 
In particular, only two components have to be synchronized: the state of the abort criterion $c$ and $d$, as well as the Schreier structure $S$ and $T$ which is manipulated by the sifting procedure. 
While the former is trivial, parallelization of the sifting procedure is discussed in the following section.

There is a slight technical issue we should address when running Algorithm~\ref{alg:random_auto} in parallel. 
If, say, the elements that are already generated by $S$ can be computed more quickly than those that are not, using many threads would create a bias towards finding the former type of element first, leading to an incrementation of $c$ with a probability larger than~$1/2$. This would break the error bound.
However, there is a simple way to fix the issue: whenever $c$ exceeds $d$, it suffices to additionally ensure all threads finish their current iteration. 

\subsection{Sifting in Parallel} \label{subsec:concurrentschreier}
For the abort criterion of the algorithm, we check whether an automorphism is contained in the group generated by the automorphisms found so far (see Line~\ref{line:c:increment} of Algorithm~\ref{alg:random_auto}). To check this, we sift it into a Schreier structure using a base of the automorphism group. 

As it turns out, sifting elements is sometimes expensive: using a conventional, sequential implementation of the sifting procedure to determine the abort criterion of our algorithm does not scale with more threads. In practice, sifting would often become the bottleneck.

For the random abort criterion we observe the following when sifting elements. 
\begin{enumerate}
    \item The base is never changed or extended. 
    \item Changes in the transversal tables $T$ are always local to one level in the Schreier structure.
	\item In practice, if sifting is expensive, many elements are sifted. 
	The computationally expensive part is then mostly multiplication of elements (Line~\ref{line:sift:multiply} of Algorithms~\ref{alg:sift} and~\ref{alg:parasift}).
\end{enumerate}
We should stress that in particular, (1) and (2) are generally \emph{not true} when sifting is employed by previous IR algorithms, and are indeed specific to the way it is used by Algorithm~\ref{alg:random_auto}. 

Crucially, these three observations enable a rather simple modification to the algorithm: we can sift elements into a Schreier structure concurrently, as long as we synchronize changes to transversal tables when changing a level. We add a lock for every level and one global lock for the generating set to enable parallel sifting on a fixed base (see Algorithm~\ref{alg:parasift}). 
\begin{algorithm}[t]
	\SetAlgoLined
	\SetAlgoNoEnd
	\caption[Sifting Procedure]{Thread-safe Sifting} \label{alg:parasift}
	\Fn{\Sift{S, T, B, $\varphi$}}{
		\SetKwInOut{Input}{Input}
		\SetKwInOut{Output}{Output}
		\Input{generators $S$, transversal table $T$, base $B$, element $\varphi$}
		\Output{whether $S$ and $T$ remained unchanged}
		\For{$i = 1;\ i \leq |B|;\ i = i + 1$}{
			$b_i$ := $\varphi(B_i)$\;
			$t$ := $(T_i)_{b_i}$\;
			\lIf{$t = \bot$}{\Break
			}
			$\varphi$ := $\varphi \cdot t^{-1}$\;
		}
		\If{$\varphi \neq \id$}{
			{\textbf{\textit{acquire lock for level $i$}}}\;
			{\textbf{\textit{acquire lock for generators}}}\;
			$S$ := $S \cup \{\varphi\}$\;
			{\textbf{\textit{release lock for generators}}}\;
			$b_i$ := $\varphi(B_i)$\;
			update $(T_i)_{b_i} = \varphi$\;
			{\textbf{\textit{release lock for level $i$}}}\;
			\Return{false}\;
		}
		\Return{true}\;
	}
\end{algorithm}
We should remark that the locking mechanism could be made more granular to further improve scaling.
However, due to observation (3), we never deemed this necessary in practice.
\subsection{Uniform Pruning} \label{subsec:uniformpruning}
\begin{figure}
	\centering
	\begin{tabular}{c c c}
		\begin{minipage}{.25\textwidth}
			\centering
			\noindent \scalebox{0.75}{\begin{forest}
				for tree={grow=south,circle,draw,edge={thick,->,xarrow,color=black},thick,minimum size=12pt,inner sep=0.5pt}
				[,fill=orange,
				[$\scriptstyle\mu$,fill=orange, [,fill=orange,[$\scriptstyle\tau_1$,fill=orange] [$\scriptstyle\tau_2$]] [[$\scriptstyle\gamma_1$] [$\scriptstyle\gamma_2$]]]
				[$\scriptstyle\mu'$ [[$\scriptstyle\tau_3$] [$\scriptstyle\tau_4$]] [[$\scriptstyle\gamma_3$] [$\scriptstyle\gamma_4$]]]
				]
			\end{forest}}
		\end{minipage}& \quad\quad\quad\quad\quad &
		\noindent\begin{minipage}{.25\textwidth}
			\centering
\noindent \scalebox{0.75}{\begin{forest}
				for tree={grow=south,circle,draw,edge={thick,->,xarrow,color=black},thick,minimum size=12pt,inner sep=0.5pt}
				[,fill=orange,
				[$\scriptstyle\mu$,fill=orange, [,fill=orange,[$\scriptstyle\tau_1$,fill=orange] [$\scriptstyle\tau_2$]] [[$\scriptstyle\gamma_1$] [$\scriptstyle\gamma_2$]]]
				[$\scriptstyle\mu'$ [[$\scriptstyle\tau_3$] [$\scriptstyle\tau_4$]] [,pruned]]
				]
			\end{forest}}
		\end{minipage}
	\end{tabular}
	\caption{Example search tree illustrating non-uniform pruning. A search tree before (left) and after (right) pruning is shown. Orange nodes indicate the path of the target leaf $\tau_1$. Leaves $\tau_i$ indicate occurrences of $\tau_1$, $\gamma_i$ indicates occurrences of $\gamma_1$. In the pruned tree, an occurrence of $\tau$ is found more likely in the subtree of $\mu'$ than $\mu$ (through random walks, assuming the immediate path to the pruned node is not taken).}
	\label{fig:nonuniformpruning}
\end{figure}
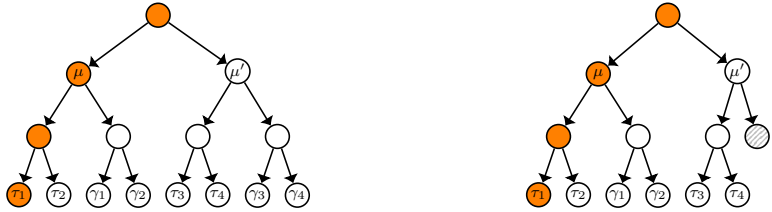

While Algorithm~\ref{alg:random_auto} is able to solve the problem on its own and parallelizes whenever enough random walks are required, it never actually prunes the search tree. 
This means that if it is discovered during random probing that a certain path does not lead to an occurrence of the target leaf, there is no mechanism to prevent making the same bad choices again. 

To rectify this, we want to intersperse the random walks of the tree with breadth-first search. 
Fortunately, probing after breadth-first traversal of an entire level has been performed naturally results in a uniform distribution of leaves again. Indeed,
upon completion of an entire level of breadth-first traversal, probing can also be characterized as starting a random walk at a uniformly random node of that level. 
We probe from a level $k$ by choosing uniformly at random a node $\nu'$ with $|\nu'| = k$ of the search tree that satisfies $\Inv(\nu') = \Inv(\nu)$. Here $\nu$ is the prefix of length $k$ of the vertex sequence corresponding to the target leaf. 
Due to Lemma~\ref{lem:tree_invariant}, the trees rooted in prefixes that contain some occurrence of the target leaf are isomorphic. 
Therefore they yield the same probability for finding an occurrence of the target leaf within them.
The process therefore samples automorphisms with a uniform distribution. Hence, breadth-first traversal can safely be combined with Algorithm~\ref{alg:random_auto}.

When using the breadth-first traversal, we still want to be able to use the automorphism pruning rule $\PruB$. However, this can break uniformity. Assume for example $\nu$ and $\nu'$ are both children of $\mu'$. Assume further they each correspond to a ``bad choice'' on this level of the search tree in the sense that neither of them contains an occurrence of the target leaf. During the algorithm we would not yet know whether these choices are bad, but suppose we find an automorphism mapping one to the other. By contracting them to a single node, we reduce the number of children without a target leaf and thus increase the chance of finding an occurrence of the target leaf in $\mu'$ (as illustrated by Figure~\ref{fig:nonuniformpruning}). 

Our solution to this problem is that whenever we use $\PruB$, we artificially restore uniformity. We do so by introducing \emph{weights} to nodes of the tree, which essentially denote the number of paths represented by a node of the tree. When combining elements of the same level, their weights (represented paths) are combined as well. Later on, when probing for leaves, we take weights into account when sampling random elements. Considering the example again, when we contract the bad choices $\nu$ and $\nu'$ to a single node, and both have weight $1$ to begin with, the remaining node gets weight $2$. The remaining node is then chosen with the same probability as both of the initial nodes together --- hence, keeping the same probability of finding the target leaf in $\mu$.

Since the probability of finding automorphisms is supposed to remain constant, one might wonder why this kind of pruning should be applied at all. The reason is that the work for the breadth-first traversal is reduced: we need to compute less nodes when advancing the breadth-first level, since symmetric nodes are contracted. On later levels, we may be able to throw away nodes (uniformly, since we are performing breadth-first traversal) and thus actually increase the probability of finding target leaves.

We now formalize the notion of weights. We describe this using the following construction: we start by defining internal weights $\overline{w}$ and external weights $w$ for all nodes. The internal weights and external weights of all nodes are initially set to 1, i.e., $w(\nu) := \overline{w}(\nu) := 1$ for all $\nu$. The internal weights are then manipulated by the algorithm. Whenever internal weights are modified, the following formula calculates the corresponding external weights:
\[w(\nu) := \begin{multicases}{2} 1 &\text{ if } \nu = \epsilon\\
\overline{w}(\nu) \cdot w(\nu_1, \dots{}, \nu_{k-1}) & \text{ if } \nu = \nu_1, \dots{}, \nu_k
\end{multicases}\]

\noindent We now modify $\PruB$ into $\PruBW$ by making it update the internal weights in addition to pruning: assume we already have $\varphi$ of $\Aut(G)$ available. For all nodes $\nu$ where $\nu^\varphi$ is not a prefix of the target leaf, $\PruBW(\nu, \nu^\varphi)$ removes $\nu^\varphi$ from the search tree and updates $\overline{w}(\nu) = \overline{w}'(\nu) + \overline{w}'(\nu^\varphi)$, where $\overline{w}'$ are the previous internal weights.

Additionally, we want to formalize the notion that the tree rooted in $\nu^\varphi$ is now \emph{represented} by the tree rooted in $\nu$. We do this by introducing an equivalence relation $\sim$, which we update every time $\PruB$ is executed. Initially, every node represents itself, hence $[\nu]_{\sim} = \{\nu\}$ holds. Note that trivially $|[\nu]_{\sim}| = w(\nu) = 1$ is satisfied initially. We update the relation whenever $\PruBW(\nu, \nu^\varphi)$ is executed. The old relation $\sim'$ is then replaced by $\sim$, which we define in the following. We do so using three states of the search tree: first, the \emph{unpruned tree} $\T(G, \nu)$, which is the initial tree where no pruning rule has been applied. Secondly, there is the \emph{pruned tree before} the operation $\PruB'$ was executed, i.e., $\T(G, \nu)'$ with the relation $\sim'$. Lastly, there is the \emph{pruned tree after} the operation $\PruB'$ has taken place, for which we want to define the relation $\sim$. 
The goal is then to argue inductively that the size of the equivalence class is equal to the external weight of the representative in the pruned trees. 

We stipulate that $\nu^\varphi$, as well as all nodes currently represented by $\nu^\varphi$, are now represented by $\nu$. Formally, we unify the equivalence classes of the root nodes in question, i.e., we set $[\nu]_{\sim} := [\nu]_{\sim'} \cup [\nu^\varphi]_{\sim'}.$

\noindent We extend this definition recursively for all nodes of the tree rooted in $\nu$, i.e., all elements of the unpruned tree rooted in $\nu^\varphi$ need to be represented by some element of $\nu' \in \T(G, \nu)'$. The tree rooted at $\nu^\varphi$ may also represent other trees, which need to be consider. We let $[\nu^\varphi]_{\sim'} = \{\nu^{\varphi_1}, \dots{},\nu^{\varphi_k}\}$, where clearly $\varphi_i \in \langle S \rangle$ for $i \in \{1, \dots{}, k\}$ holds. Since all of these trees may have been pruned differently through applications of $\PruBW$, we refer back to nodes of the unpruned tree $\T(G)$. For every node $\nu'$ in the pruned tree $\T(G, \nu)'$, we find all of the nodes of the unpruned tree that are now represented by it, i.e., 
$
[\nu']_{\sim} := [\nu']_{\sim'} \cup
\{\nu''^{\varphi_i} \; | \; \nu'' \in [\nu']_{\sim'}, i \in \{1, \dots{}, k\} \; | \; \nu'' \in \T(G, \nu)\}
$.
Note that this is well-defined in terms of equivalence relations, since all $\varphi_i$ define bijections between nodes of $\T(G, \nu)$ and $\T(G, \nu^{\varphi_i})$ (Lemma~\ref{lem:auto_tree_correspondence1}). All other nodes of the pruned tree keep their equivalence classes of $\sim'$, which is correct since their weight is unaltered as well.

We now argue that the external weight of a remaining node represents the size of its respective equivalence class:

\begin{lemma} \label{lem:weight_class_correspondence} Let $A$ be a sequence of applications of $\PruBW$ to a search tree $\T(G)$. Assume that applying $A$ results in $\T(G)'$ with external weight function $w'$ and $\sim'$ is the corresponding equivalence relation (as described previously). Then, it holds that $|[\nu]_{\sim'}| = w'(\nu)$.
\end{lemma}
\begin{proof} We argue the claim by induction over the sequence $A$. Initially, the claim is true since all weights are $1$ and every node represents itself in $\sim$.

Let us now argue that when applying some $\PruBW(\nu, \nu^\varphi)$, the invariant remains valid. For the root node this is easy to see: by induction $|[\nu]_{\sim'}| = w'(\nu)$ and $|[\nu^\varphi]_{\sim'}| = w'(\nu^\varphi)$ hold, which in turn shows $|[\nu]_{\sim}| = |[\nu]_{\sim'} \cup [\nu^\varphi]_{\sim'}| = w'(\nu) + w'(\nu^\varphi) = w(\nu)$.

\noindent The third equality holds since $\varphi \neq \id$ is required by definition. We now show this for all $\nu' \in \T(G, \nu)'$ by induction. Let $\nu' = \nu'_1, \dots{}, \nu'_k$ and $\mu = \nu'_1, \dots{}, \nu'_{k-1}$. We know that $w(\mu) = |[\mu]_{\sim}|$ holds, i.e., the statement is true for the parent of $\nu'$. 

It can be easily seen that set of the elements $\mu' \in \T(G, \nu)$ with $\mu' \in [\mu]_\sim$, i.e., elements represented by $\mu$ in the subtree of $\nu$, has not been altered. Hence, we can rewrite $w(\mu) = |[\mu]_\sim| = |\{\mu' \in [\mu]_\sim \; | \; \mu' \in \T(G, \nu)\}| \cdot w(\nu).$

\noindent The internal weight $\overline{w}$ of $\nu'$ is only changed whenever $\PruBW$ is directly applied on $\nu'$. This means that in the unpruned tree, $\nu'$ represents $\overline{w}(\nu')$ many elements in $\T(G, \mu)$. We can conclude
$
|[\nu']_{\sim}| = |\{\nu'' \in [\nu']_{\sim'} \; | \; \nu'' \in \T(G, \nu)\}| \cdot w(\nu)
= \overline{w}(\nu') \cdot |\{\mu' \in [\mu]_\sim \; | \; \mu' \in \T(G, \nu)\}| \cdot w(\nu) 
= \overline{w}(\nu') \cdot w(\mu),
$
which proves our claim.
\end{proof}

\noindent Since we weigh each equivalence class of nodes with its size in the pruned tree, it does not matter up to $\sim$ whether we perform random walks on the pruned search tree or the unpruned tree. The distributions of equivalence classes are indistinguishable. Let us now argue why this suffices for the correctness of our algorithm. By looking carefully at the previous discussion, we can observe that all elements of an equivalence class can be mapped to each other through elements of $\langle S \rangle$. This implies that given a leaf $\nu$, all elements represented by $\nu$ are generated by $S$ if and only if $\nu$ is generated by $S$. In terms of Lemma~\ref{lem:group_abort_criterion2}, the automorphism derived from a leaf sampled uniformly at random from $\nu$ therefore has the same chance of sifting through the structure as one of $\nu^\varphi$. Therefore, by using a modified version of Lemma~\ref{lem:group_abort_criterion2} as the abort criterion, we can consider sampled weighted nodes of the pruned tree instead of proper uniform random nodes of the unpruned tree:

\begin{lemma} \label{lem:group_abort_criterion3} Let $G$ be a graph, $B$ a base, $S$ a set of permutations, $\T(G)'$ the search tree resulting from repeated application of $\PruBW$ with elements of $S$ and $\tau \in \T(G)'$ the target leaf. Furthermore, let $\nu$ be a leaf drawn from $\T(G)'$ with weight $w(\nu)$ where $\nu^{-1} \cdot \tau = \varphi \in \Aut(G)$. If $\langle S \rangle \neq \Aut(G)$ holds, then the probability that $\varphi$ does not sift through the Schreier structure defined by $B$ and $S$ is at least $\frac{1}{2}$.
\end{lemma} 
\begin{proof} 
Since $\nu$ is an occurrence of the target leaf, we can require $\nu^{\varphi} = \tau$. From the previous discussion, we know that external weights of remaining leaves determine the amount of leaves in the unpruned tree represented by them (Lemma~\ref{lem:weight_class_correspondence}) and that no other remaining leaves represent them. It therefore suffices to argue that the derivable automorphisms of leaves of $\T(G)$ represented by $\nu$ all have the same chance of being generated by $\langle S \rangle$. Since leaves are represented by $\nu$ if they were pruned using $\PruBW$ -- which by assumption can only use elements of $S$ -- all of them can be derived by applying some $\varphi' \in \langle S \rangle$ to $\nu$: hence, $\varphi \in \langle S \rangle \iff \varphi \cdot \varphi' \in \langle S \rangle$.
\end{proof}

\noindent The actual solver will --- in addition to finding a generating set in the first place --- still need to fill the Schreier structure sufficiently. This can however only require more elements and thus increases the chance of elements not sifting successfully. Hence, it even decreases error probabilities.

\section{Implementation and Heuristics} \label{sec:implementation}
We use the tools of the previous section to construct the parallel graph automorphism solver \dejavu{}. We start by describing the high-level structure of \dejavu{}.
It consists of $4$ modes of operation, between which it continuously switches. 
The solver decides how to switch between the modes using heuristics, which are based on a cost estimation.

The solver begins by trying to sample a good cell selector (in parallel). Then, \emph{base-aligned} automorphism search is performed (Section~\ref{subsec:autosample}), followed by breadth-first search (Section~\ref{subsec:bfs}), followed by \emph{level} automorphism search (see Section~\ref{subsec:autosample}). 
At any point, depending on the cost estimation, the solver can decide to go to a preceding mode again.

\noindent The tool is written in \texttt{C++}. Threads of the \texttt{C++} standard library are used for parallelization. It contains some modified code of the \nauty{} / \Traces{} distribution (available at \cite{nautyTracesweb}), namely specialized versions of the Schreier-Sims implementation. The source code is available at \cite{dejavuweb}.

The refinement of \dejavu{} is a highly-engineered version of color refinement \cite{McKay201494}, following the implementation of \Traces{} closely. 
We also exploit the blueprint technique as described in \cite{alenextoappear},
extending the technique to also cache cell selector choices for subsequent branches.

Next, we provide further details on practical and conceptual aspects of the solver. 

\subsection{Breadth-first Traversal and Trace Deviation Sets} \label{subsec:bfs} 
The work of the breadth-first traversal is shared between threads through lock-free queues. A master thread adds all of the elements which have to be computed to a queue. Threads then dequeue a chunk of work, compute the elements, and report their results back to the master thread through another queue.
 In order to minimize overhead, large chunks of work are enqueued and dequeued from the queue rather than single elements.
Furthermore, \dejavu{} uses automorphism pruning when performing breadth-first search as described in Section~\ref{subsec:uniformpruning} while filling up the queue (not enqueuing multiple elements known to be isomorphic).

During breadth-first traversal, we make use of a \emph{trace} invariant, as introduced by \Traces{} and described in \cite{McKay201494}.
Furthermore, we introduce the \emph{trace deviation set} technique.
The approach is related to the special automorphism algorithm of \Traces{} \cite{McKay201494} as well as the trace deviation trees used in~\cite{alenextoappear}.

\begin{figure}
	\centering
	\noindent \scalebox{0.75}{\begin{forest}
		for tree={grow=south,circle,draw,edge={thick,->,xarrow,black},thick,minimum size=12pt,inner sep=0.5pt}
		[,fill=orange,
		[,fill=orange[,fill=orange,] [] [$\scriptscriptstyle d1$,pruned] [$\scriptscriptstyle d2$,pruned]]
		[,pruned,[$\scriptscriptstyle d3$,pruned]]
		[[] [] [$\scriptscriptstyle d1$,pruned] [$\scriptscriptstyle d2$,pruned]]
		[,pruned, [] [] [] [$\scriptscriptstyle d1$,pruned]]
		]
	\end{forest}}
\caption{Potential search tree traversed when using trace deviation sets. Orange indicates base nodes and gray indicates pruned nodes. Children of pruned nodes are of course eventually pruned as well.}
\label{fig:deviation_search}
\end{figure}
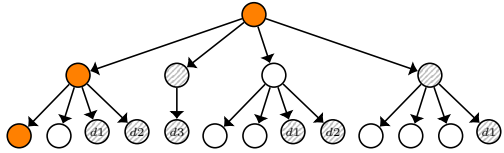

During breadth-first traversal, we also keep a \emph{trace deviation set}. 
The idea of this pruning technique is related to the special automorphism algorithm of \Traces{} (see \cite{McKay201494}) and the trace deviation tree technique of \cite{alenextoappear}. 
We present the idea in the following.

To describe the technique, we first define a new node invariant, which we call the \emph{deviation value} $\Dev_{\Inv}\colon V^* \to \mathbb{N}^2 \cup \{\bot\}$. 
Consider a fixed trace $\Inv(\tau)$, which for our purposes will be the trace invariant of the target leaf $\tau$. 
The deviation value $\Dev_{\Inv}(\nu)$ for a node $\nu$ is then defined as a tuple of the first position and the corresponding value in the trace $\Inv(\nu)$ that is different from $\Inv(\tau)$. 
If there are no differences, we set the deviation value to $\bot$ denoting ``no deviation''. 
Since the deviation value is a function of the invariant computed up until an isomorphism-invariant point, it is also naturally invariant under isomorphism.

Consider a node $\mu$ in the search tree. 
The crucial observation is that in our algorithm, we can also use the set of deviation values of its children as an invariant for $\mu$ itself. 
Assume $\nu_1, \dots{}, \nu_k$ are children of $\mu$ and none of the subtrees rooted in the children has been pruned through $\PruA$. Then, $D(\mu) := \{\Dev_{\Inv}(\nu_1), \dots{}, \Dev_{\Inv}(\nu_k)\}$, the trace deviation set of $\mu$, can be used as an invariant for $\mu$: we claim that for any other node $\mu'$ with children $\nu_1', \dots{}, \nu_k'$, it must hold that
$
	D(\mu) = \{\Dev_{\Inv}(\nu_1), \dots{}, \Dev_{\Inv}(\nu_k)\} = \{\Dev_{\Inv}(\nu_1'), \dots{}, \Dev_{\Inv}(\nu_k')\} = D(\mu')
$
whenever $\mu$ and $\mu'$ are isomorphic. 
If no pruning has taken place, this is easy to see since the branches are isomorphic by assumption, immediately implying that branches must contain the same invariant values. 
But if $\PruBW$ is applied, no invariant values can be removed either: since pruned nodes are isomorphic to remaining nodes, they, again immediately by definition, must contain the same invariant values. 

When advancing in a breadth-first manner, the aforementioned requirements are guaranteed to be satisfied: no $\PruA$ has taken place on the level that is currently being pruned. 
Furthermore, while computing the level, the set of deviation values is automatically calculated anyway: whenever we observe that a node $\nu$ below $\mu$ deviates from the desired invariant $\Inv(\nu)$ and should be pruned using $\PruA$, we already have enough information to derive $\Dev_{\Inv}(\nu)$.

These observations are specifically exploited as follows: first, all children of the base node $\tau'$ (which belongs to the path on the way to the target leaf $\tau$) are computed. 
If nodes deviate from the trace, their deviation values are recorded into a set, i.e., we calculate the trace deviation set $D(\tau')$. 
The idea is that if a node is (supposedly) isomorphic to the base node $\tau'$, then, for its (supposedly) isomorphic children, it must deviate from the trace in the same manner at the same position. 
Hence, for all other parent nodes $\mu$, we also keep track of $D(\mu)$ when calculating their children. 
Whenever we discover a new element of $D(\mu)$, we check whether the set equivalence $D(\tau') = D(\mu)$ can still be satisfied. 
If not, $\mu$ can be pruned immediately, without the necessity to calculate all its children.

For example, assume that we calculated deviation values $\{d_1, d_2\}$ for the base node (illustrated in Figure~\ref{fig:deviation_search}). 
From the previous discussion, it follows that we can immediately prune all nodes that produce a value other than $\{d_1, d_2\}$. 
We can also prune nodes that do not produce all of the deviation values. 
If for example a value $d_3 \notin \{d_1, d_2\}$ is encountered, the parent node can immediately be removed from the tree.

A crucial point is that pruning through trace deviation sets has negligible cost: children of the base node always have to be computed, and the trace deviation does not necessitate more calculation than is done for that particular node anyway. 
We are still able to fully use the early-out capabilities of the trace invariant. 
In the parallel implementation, only the initial recording of deviation values of base nodes has to be synchronized. 
In the implementation we do however, depending on a heuristic, use a slight variation: to make deviation values more distinct, it is sometimes beneficial to not use the early-out immediately. 
Instead, for a fixed constant $k$, color refinement is continued past the deviation for $k$ more cells, accumulating more information for the deviation value. 
The trade-off is as follows: if $k$ becomes larger, the early-out in color refinement is taken later, but deviation values become more distinct. 
In practice, this trades per-node cost for the number of nodes in the search tree. 
However, in our experiments we observed that even for very small $k$ (we use $k = 5$), node reduction can be substantial -- while not increasing per-node cost by a significant amount.

\subsection{Automorphism Sampling} \label{subsec:autosample}
\begin{figure}
\centering
\scalebox{0.7}{
\includegraphics[clip]{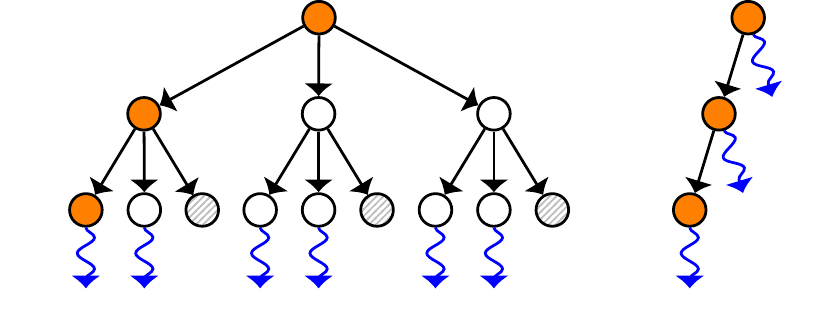}}
    \caption{Level search (left) and base-aligned search (right). Squiggly lines indicate where random walks originate, orange indicates base nodes and gray nodes are pruned by breadth-first traversal.}
    \label{fig:search_styles}
\end{figure}   
A central aspect when sampling automorphisms is whether we can guarantee that the resulting automorphisms are distributed uniformly in the group or not. If the distribution is uniform, they count towards the probabilistic abort criterion of the solver. 
The solver uses two approaches, namely \emph{base-aligned search} (generally non-uniform) and \emph{level search} (uniform).

Level search is essentially sampling as described in Section~\ref{sec:randomir} (see Figure~\ref{fig:search_styles}, left). Walks are initiated from the remaining nodes of the search tree at the current breadth-first level, and when drawing the initial node, weights are accounted for to make the resulting distribution of automorphisms uniform.
Crucially, in this mode, \dejavu{} usually stores additional target leaves, which is proven to result in an exponential speedup in the worst-case \cite{DBLP:conf/icalp/AndersS21, alenextoappear}. 

Base-aligned search is designed to find as many automorphisms as possible with minimal effort. 
This typically entails giving up on uniformity. 
Base-aligned search initiates random walks from a base point of a given strategy (see Figure~\ref{fig:search_styles}, right).
Whenever finding automorphisms from a certain base point is deemed hard, search is advanced to the next base point. 
As a side effect, this handles inherently easy graphs efficiently: whenever color refinement already determines the orbit partition, base-aligned search finds all automorphisms and is even able to terminate search deterministically.


\section{Practical Performance} \label{sec:performance}
We provide benchmark results corroborating that \dejavu{} performs competitively on a large variety of graphs while scaling with the use of more threads. 
\subsection{Benchmarks} \label{subsec:benchmarks}
All benchmarks were run on an Intel Core i7 9700K (8 cores) processor with 16GB of RAM and Ubuntu 19.04. All graphs were randomly permuted, but every solver was given the same permuted version of a graph. All runtimes are measured without parsing the input. 

The benchmarks include most sets from \cite{nautyTracesweb}, which is the de facto standard when it comes to symmetry computation. 
We extended two of the sets to larger instance sizes. 
The respective graphs can be found in \cite{dejavuweb}. 
We should point out that for random trees and pipe graphs \Traces{} benefits from specialized code that is not implemented in \dejavu{}.

The (user-definable) error bound for \dejavu{} was set to $1\%$. A $1\%$ error bound means that with at most that probability at least $1$ generator of the generating set is missing. 
It can be proven however that the probability of 
missing at least $2$ generators then only has a probability of at most $\frac{1}{2}\%$, missing at least $3$ only at most $\frac{1}{2^2}\%$ and so forth (the argument for this is similar to the proof of Lemma~\ref{lem:error:prob:of:algo} involving the index of the subgroup found).
Hence, even if errors occur, it is highly likely only a small part of the generating set is missing.

Actual error probabilities are even lower. Due to the one-sided nature of the error, on asymmetric graphs \dejavu{} can not err (e.g. most random regular graphs, random graphs, multipedes, \texttt{latin-sw}). 
Secondly, on many graph families \dejavu{} can invoke a deterministic criterion for termination (e.g. for most of \texttt{rantree}, \texttt{hypercubes}, \texttt{dac}, \texttt{lattice}, complete graphs, \texttt{tran}). 
This means for the majority of the benchmarks errors can not be observed.

\begin{figure}
	\centering
	\includegraphics[scale=0.9]{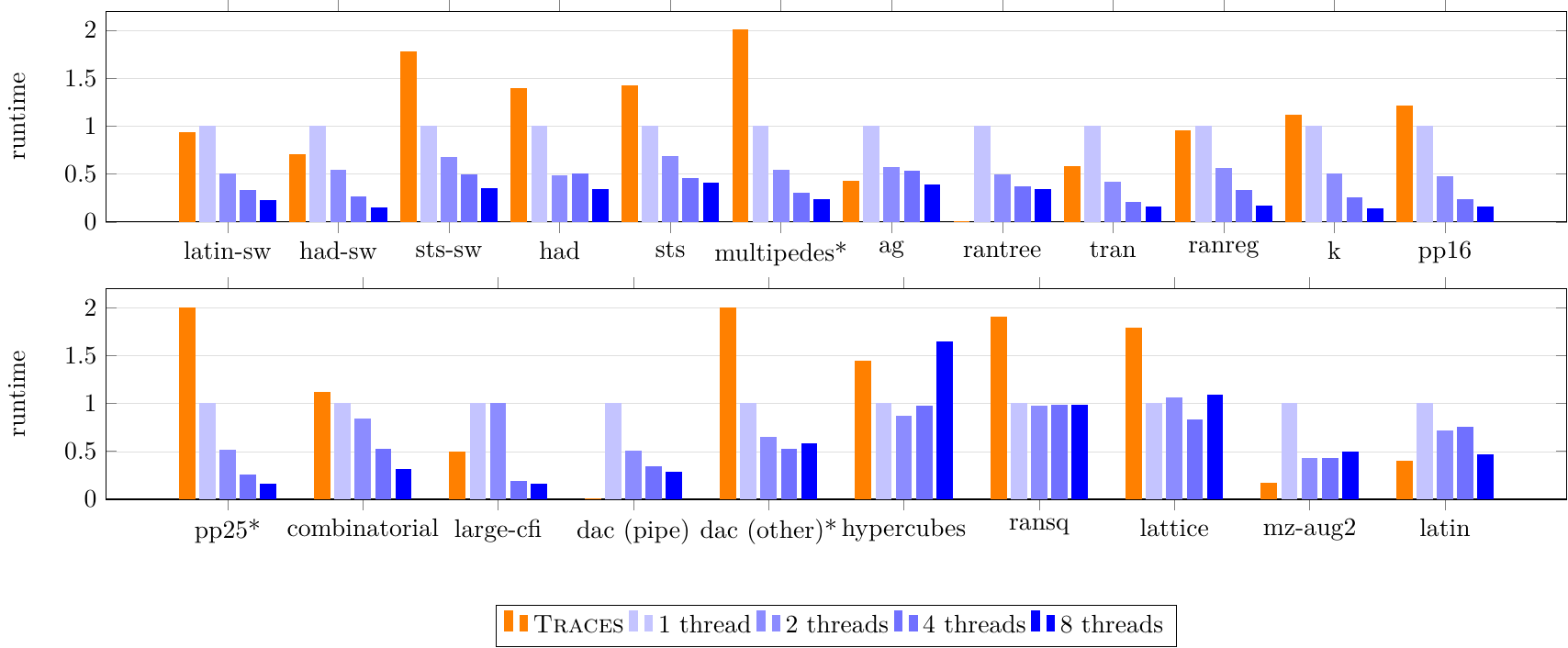}
\caption{Relative runtimes of benchmark sets for \Traces{} and \dejavu{} (using the respective number of threads). For sets marked with *, we capped the runtime of \Traces{} at $2$.} \label{fig:scaling}
\end{figure}

\newcommand{\figureScaleS}{0.78}

The benchmarks corroborate that through the use of parallelism, we are able to achieve significant speedup over state-of-the-art tools in a domain that so far has been exclusively dominated by sequential solvers (see Figure~\ref{fig:scaling}). Specifically, we do so on the particular, representative benchmark suite state-of-the-art solvers are tuned to solve.

Overall benchmarks show that \dejavu{} with a single thread performs competitively with \Traces{}
(on $14$ benchmark sets \dejavu{} is faster, while \Traces{} is faster on the other $10$ sets).
Using $8$ threads, \dejavu{} outperforms \Traces{} on most sets (on $19$ out of $24$ sets).
Additionally, on \texttt{lattice}, \dejavu{} performs better than \Traces{} for larger instances.
Figure~\ref{fig:first_experiments} and Figure~\ref{fig:all_experiments} shows detailed plots for the tested benchmark sets.
Note that the random Erd\H{o}s-Rényi graphs tested in Figure~\ref{fig:first_experiments} exhibit no scaling whatsoever, as they are entirely solved by color refinement for larger instances. For the sake of clarity, with the exception of \texttt{ransq}, they thus do not appear in Figure~\ref{fig:scaling}.

\pgfplotsset{grid=both,major grid style={black!20},minor grid style={draw=none},minor tick style={draw=black!20}}

\newcommand{\figureScale}{0.5}
\newcommand{\figureScaleX}{0.8}

\begin{figure}
	\begin{subfigure}{.45\textwidth}
	\centering
\includegraphics[scale=\figureScaleS]{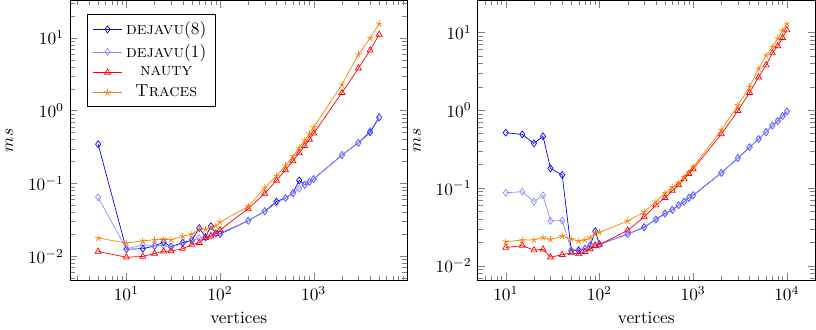}
\caption{Random graphs with edge probability $\frac{1}{2}$ (left) and $\frac{1}{10}$ (right).\\\hspace{1cm}} \label{fig:bench1}
\end{subfigure}\quad\quad
\begin{subfigure}{.45\textwidth}
	\centering
	\includegraphics[scale=\figureScaleS]{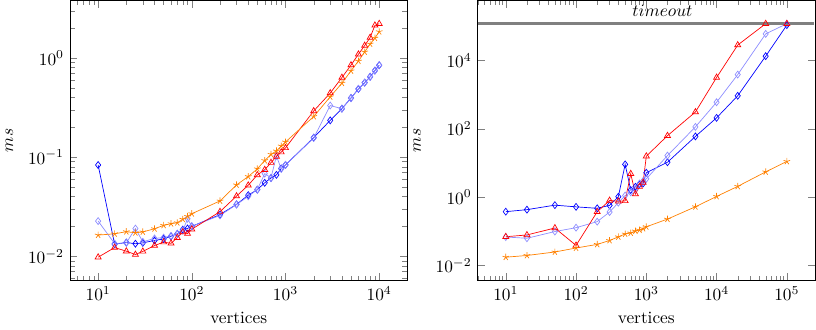}
\caption{Random graphs with edge probability $\frac{1}{\sqrt{n}}$ (left) and random trees (right), timeout is $120$ seconds.} \label{fig:bench3}
\end{subfigure}
\centering
\caption{Detailed plots for benchmarks.} \label{fig:first_experiments}
\end{figure}

\begin{figure}[p!]
\begin{subfigure}{.45\textwidth}
	\centering
	\includegraphics[scale=\figureScaleS]{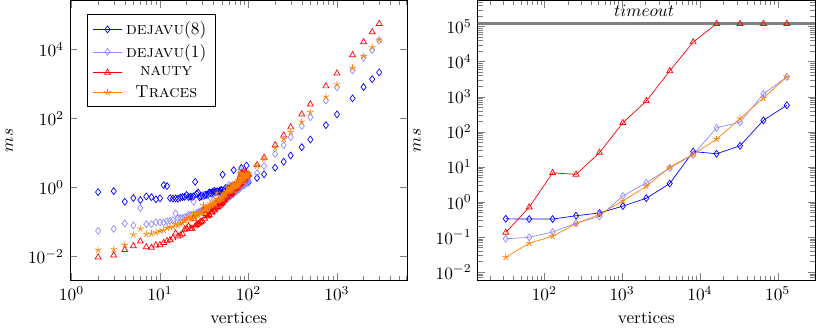}
\caption{Complete graphs (\texttt{k}, left) and random 6-regular graphs (\texttt{ranreg}, right), timeout is $120$ seconds.} \label{fig:bench4}
\end{subfigure}\quad\quad
\begin{subfigure}{.45\textwidth}
	\centering
	\includegraphics[scale=\figureScaleS]{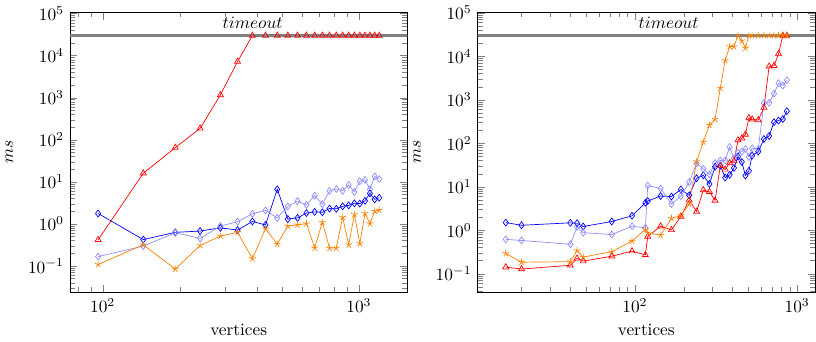}
	\caption{Miyazaki graphs (\texttt{mz-aug2}) and shrunken multipedes (very small instances), timeout is $30$ seconds.} \label{fig:bench5}
\end{subfigure}
\begin{subfigure}{.45\textwidth}
	\centering
	\includegraphics[scale=\figureScaleS]{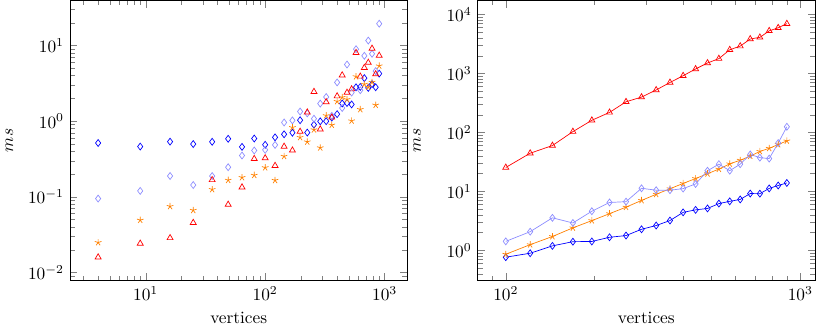}
\caption{Latin squares (\texttt{latin}, left) and switched edge variants (\texttt{latin-sw}, right).} \label{fig:bench7}
\end{subfigure}\quad\quad
\begin{subfigure}{.45\textwidth}
	\centering
	\includegraphics[scale=\figureScaleS]{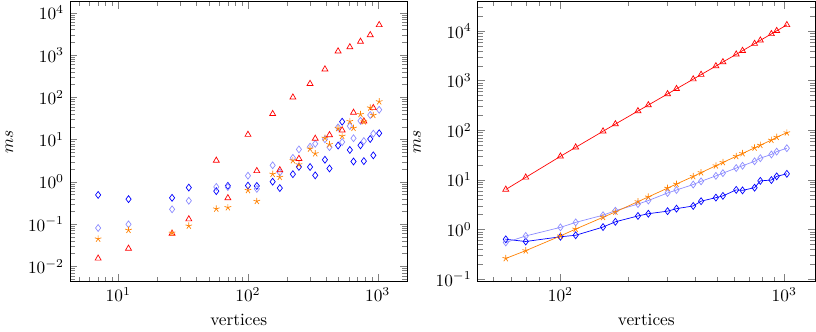}
\caption{Steiner triple systems (\texttt{sts}, left) and switched edge variants (\texttt{sts-sw}, right).} \label{fig:bench8}
\end{subfigure}
\begin{subfigure}{.45\textwidth}
	\centering
	\includegraphics[scale=\figureScaleS]{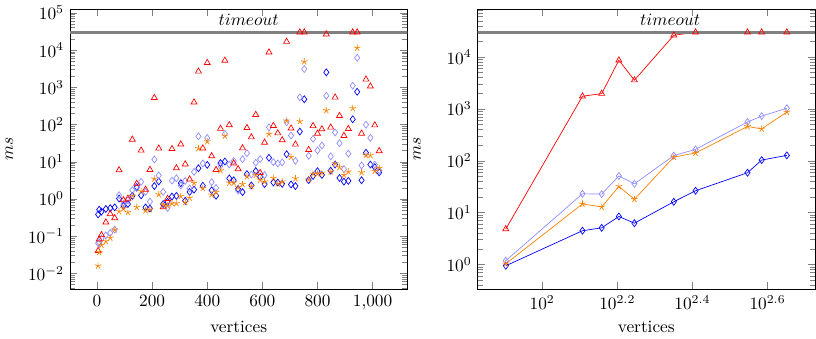}
\caption{Hadamard matrices (\texttt{had}, left) and switched edge variants (\texttt{had-sw}, right), timeout is $30$ seconds.} \label{fig:bench9}
\end{subfigure}\quad\quad
\begin{subfigure}{.45\textwidth}
	\centering
\includegraphics[scale=\figureScaleS]{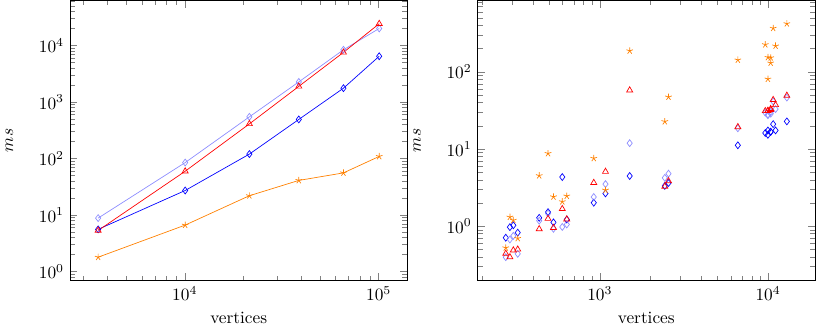}
\caption{Model graphs of CNF formulas (\texttt{dac}), split into the ``pipe'' graphs (left) and the remainder of the set (right).} \label{fig:bench10}
\end{subfigure}
\begin{subfigure}{.45\textwidth}
	\centering
	\includegraphics[scale=\figureScaleS]{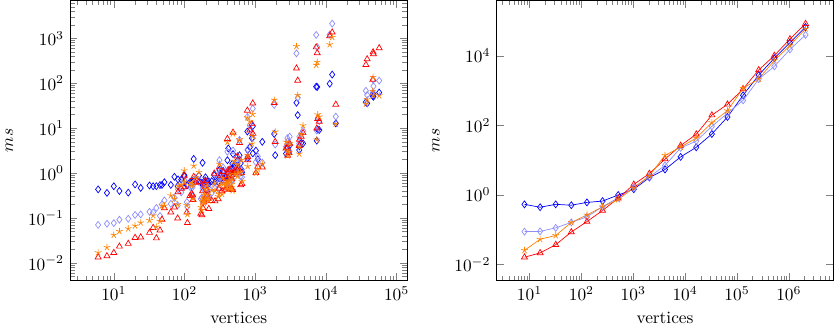}
\caption{Vertex transitive graphs (\texttt{tran}, left) and hypercubes (\texttt{hypercubes}).} \label{fig:bench11}
\end{subfigure}\quad\quad
\begin{subfigure}{.45\textwidth}
	\centering
	\includegraphics[scale=\figureScaleS]{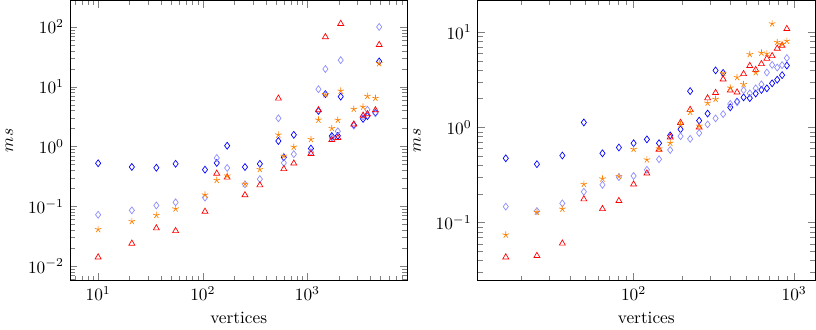}
	\caption{Affine geometry graphs (\texttt{ag}) and lattice graphs (\texttt{lattice}).} \label{fig:bench6}
\end{subfigure}
\begin{subfigure}{.45\textwidth}
	\centering
	\includegraphics[scale=\figureScaleS]{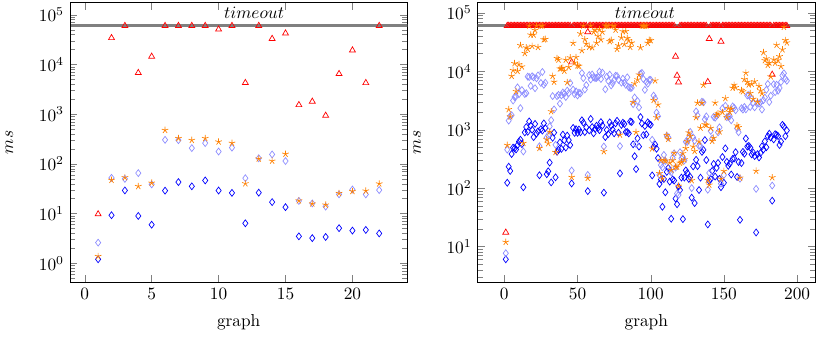}
\caption{Non-Desarguesian projective planes of order 16 (\texttt{pp16}, left) and order 25 (\texttt{pp25}, right), timeout is $60$ seconds. The x-axis arbitrarily orders the graphs in the set.} \label{fig:bench12}
\end{subfigure}\quad\quad
\begin{subfigure}{.45\textwidth}
	\centering
	\includegraphics[scale=\figureScaleS]{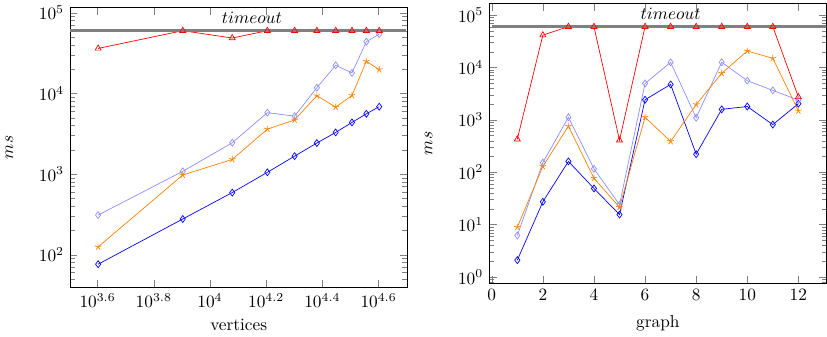}
\caption{CFI graphs (\texttt{large-cfi}, left) and a variety of graphs by Gordon Royle (\texttt{combinatorial}, right), timeout is $60$ seconds. For the latter, the x-axis arbitrarily orders the graphs in the set.} \label{fig:bench13}
\end{subfigure}
\centering
\caption{Detailed plots for benchmarks.} \label{fig:all_experiments}
\end{figure}

\subsection{Scaling} \label{subsec:scaling}
We want to discuss in detail how \dejavu{} scales with the use of more threads.
In Figure~\ref{fig:scaling}, scaling is illustrated. 
To lessen the bonus of sampling cell selectors (see Section~\ref{sec:implementation}) gained when adding threads, we provide the $1$ and $2$ thread variants with the best cell selector for each set. (Otherwise, scaling would be even better, but the comparison would not be fair.)  
The diagram shows the summed up runtimes for a graph class relative to the runtime of a single thread, i.e., the single thread variant has a runtime of $1$. 
This might overrepresent larger instances in the data, but of course larger instances are exactly those of biggest interest. 

The diagram shows that for most sets of the benchmark suite, on our hardware, \dejavu{} scales remarkably with the use of more threads. 
Overall we achieved our goal of designing a competitive tool that on most graph classes can efficiently exploit parallelism, in our benchmarks sometimes even approaching what is theoretically possible on 8 cores.

With regard to achieving perfect parallelization, there are however some caveats.  
The implementation has some aspects which do not scale with more threads. 
Typically these aspects have a negligible contribution to the overall runtime, however, on some particular instance types they do not. 
We distinguish between conceptual and technlological aspects, which we discuss in the following in order to explain the outliers in Figure~\ref{fig:scaling}.

\paragraph{Conceptual} There can be structural reasons why a work load might not scale with more threads using our approach. 
The initial color refinement and target leaf computation are sequential in our algorithm. 
Therefore, if an input is solved already by applying only color refinement, we can not expect our algorithm to scale. 
This happens on the set \texttt{ransq}.
Furthermore, if the automorphism group of a graph contains very few generators which can be found quite easily, then there is often no way for \dejavu{} to split it accordingly. The property of having an ``easy'' search tree is independent of the graph size: this can also be true on very large graphs that take a considerable amount of time. 
The \texttt{hypercubes} set of the benchmarks are a prime example for this --- they have large running times but very small search trees. 
Relative to their size, they are however easy to solve. 

\paragraph{Technological} On the technological side we observe two main sticking points of the current implementation: first of all, if a graph is solved too quickly, the various costs of using threads (startup, communication, cache performance, ...) can not be amortized, which is of course to be expected.
This causes the unfavorable scaling in Figure~\ref{fig:scaling} for the \texttt{lattice} and \texttt{mz-aug2} sets. 
They consist of too many small work loads.
However, Figure~\ref{fig:bench6} and Figure~\ref{fig:bench5} show that with growing instance sizes scaling improves greatly for both sets.

Secondly, there can be the case of degrading memory performance when adding threads.
For our processor
we observed this for sparse graphs with more than $10^5$ vertices and few automorphisms
Again, we can observe this effect on the \texttt{hypercubes} set.


\section{Conclusion and Future Work} \label{sec:conclusion} 
We presented the new randomized, parallel algorithm \dejavu{} that computes automorphism groups and can be used to compute symmetries of combinatorial objects. Benchmarks show that \dejavu{} is competitive with state-of-the-art solvers, and parallelizes to $8$ cores remarkably well on a wide variety of instances.
In order to show further scaling in the same generic manner, we believe developing more large-scale, meaningful benchmark sets is required.
For harder instances, preliminary testing shows that the use of more cores significantly improves the runtimes even further.

In future work, we intend to improve \dejavu{} by adding dedicated subroutines to handle low degree vertices and gadget graph constructions.
On very simple graphs, color refinement is the bottleneck which seems to necessitate an efficient parallel implementation for color refinement. 

\section*{Acknowledgements}
We thank Adolfo Piperno, Brendan McKay, Tommi Junttila, and Petteri Kaski for discussions providing us with deeper insights into their isomorphism solvers.
We also want to thank our colleagues Thomas Schneider, Jendrik Brachter, and Moritz Lichter for the fruitful discussions we had on some of the topics in this paper.

\bibliography{main}{}

\begin{thebibliography}{10}

\bibitem{dejavuweb}
dejavu.
\newblock \url{www.mathematik.tu-darmstadt.de/dejavu}.

\bibitem{nautyTracesweb}
{nauty and Traces}.
\newblock \url{http://pallini.di.uniroma1.it}.

\bibitem{alenextoappear}
Markus Anders and Pascal Schweitzer.
\newblock Engineering a fast probabilistic isomorphism test.
\newblock In {\em 2021 Proceedings of the Symposium on Algorithm Engineering
  and Experiments ({ALENEX})}. {SIAM}, 2021.
\newblock to appear.
\newblock \href {https://doi.org/10.1137/1.9781611976472.6}
  {\path{doi:10.1137/1.9781611976472.6}}.

\bibitem{DBLP:conf/icalp/AndersS21}
Markus Anders and Pascal Schweitzer.
\newblock Search problems in trees with symmetries: Near optimal traversal
  strategies for individualization-refinement algorithms.
\newblock In Nikhil Bansal, Emanuela Merelli, and James Worrell, editors, {\em
  48th International Colloquium on Automata, Languages, and Programming,
  {ICALP} 2021, July 12-16, 2021, Glasgow, Scotland (Virtual Conference)},
  volume 198 of {\em LIPIcs}, pages 16:1--16:21. Schloss Dagstuhl -
  Leibniz-Zentrum f{\"{u}}r Informatik, 2021.
\newblock \href {https://doi.org/10.4230/LIPIcs.ICALP.2021.16}
  {\path{doi:10.4230/LIPIcs.ICALP.2021.16}}.

\bibitem{OpenMP}
Vijaya Balpande and Anjali Mahajan.
\newblock Article: Parallelization of graph isomorphism using {OpenMP}.
\newblock {\em International Journal of Computer Applications}, 117(8):33--41,
  May 2015.
\newblock Full text available.
\newblock URL: \url{http://dx.doi.org/10.5120/20576-2982}, \href
  {https://doi.org/10.5120/20576-2982} {\path{doi:10.5120/20576-2982}}.

\bibitem{DBLP:conf/cav/ChuJ12}
Duc{-}Hiep Chu and Joxan Jaffar.
\newblock A complete method for symmetry reduction in safety verification.
\newblock In {\em Proceedings of the 24th international conference on Computer
  Aided Verification}, volume 7358 of {\em Lecture Notes in Computer Science},
  pages 616--633. Springer, 2012.
\newblock \href {https://doi.org/10.1007/978-3-642-31424-7_43}
  {\path{doi:10.1007/978-3-642-31424-7_43}}.

\bibitem{Darga:2004:ESS:996566.996712}
Paul~T. Darga, Mark~H. Liffiton, Karem~A. Sakallah, Igor~L. Markov, and Igor~L.
  Markov.
\newblock Exploiting structure in symmetry detection for {CNF}.
\newblock In {\em Proceedings of the 41st Annual Design Automation Conference},
  DAC '04, pages 530--534, New York, NY, USA, 2004. ACM.
\newblock URL: \url{http://doi.acm.org/10.1145/996566.996712}, \href
  {https://doi.org/10.1145/996566.996712} {\path{doi:10.1145/996566.996712}}.

\bibitem{DBLP:conf/nips/GensD14}
Robert Gens and Pedro~M. Domingos.
\newblock Deep symmetry networks.
\newblock In {\em Proceedings of the 27th International Conference on Neural
  Information Processing Systems}, pages 2537--2545, 2014.
\newblock URL: \url{https://dl.acm.org/doi/abs/10.5555/2969033.2969110}.

\bibitem{DBLP:reference/fai/GentPP06}
Ian~P. Gent, Karen~E. Petrie, and Jean{-}Fran{\c{c}}ois Puget.
\newblock Symmetry in constraint programming.
\newblock In Francesca Rossi, Peter van Beek, and Toby Walsh, editors, {\em
  Handbook of Constraint Programming}, volume~2 of {\em Foundations of
  Artificial Intelligence}, pages 329--376. Elsevier, 2006.
\newblock URL: \url{http://dx.doi.org/10.1016/S1574-6526(06)80014-3}, \href
  {https://doi.org/10.1016/S1574-6526(06)80014-3}
  {\path{doi:10.1016/S1574-6526(06)80014-3}}.

\bibitem{DBLP:conf/forte/Godefroid99}
Patrice Godefroid.
\newblock Exploiting symmetry when model-checking software.
\newblock In {\em Proceedings of the {IFIP TC6 WG6.1} Joint International
  Conference on Formal Description Techniques for Distributed Systems and
  Communication Protocols ({FORTE XII}) and Protocol Specification, Testing and
  Verification ({PSTV XIX})}, volume 156 of {\em {IFIP} Conference
  Proceedings}, pages 257--275. Kluwer, 1999.
\newblock \href {https://doi.org/10.1007/978-0-387-35578-8_15}
  {\path{doi:10.1007/978-0-387-35578-8_15}}.

\bibitem{DBLP:journals/cacm/HeuleK17}
Marijn J.~H. Heule and Oliver Kullmann.
\newblock The science of brute force.
\newblock {\em Communications of the {ACM}}, 60(8):70--79, 2017.
\newblock URL: \url{https://dl.acm.org/doi/10.1145/3107239}, \href
  {https://doi.org/10.1145/3107239} {\path{doi:10.1145/3107239}}.

\bibitem{DBLP:conf/alenex/JunttilaK07}
Tommi~A. Junttila and Petteri Kaski.
\newblock Engineering an efficient canonical labeling tool for large and sparse
  graphs.
\newblock In {\em Proceedings of the Nine Workshop on Algorithm Engineering and
  Experiments, {ALENEX} 2007, New Orleans, Louisiana, USA, January 6, 2007}.
  {SIAM}, 2007.
\newblock \href {https://doi.org/10.1137/1.9781611972870.13}
  {\path{doi:10.1137/1.9781611972870.13}}.

\bibitem{DBLP:conf/sat/KatebiSM10}
Hadi Katebi, Karem~A. Sakallah, and Igor~L. Markov.
\newblock Symmetry and satisfiability: An update.
\newblock In Ofer Strichman and Stefan Szeider, editors, {\em Theory and
  Applications of Satisfiability Testing - {SAT} 2010, 13th International
  Conference, {SAT} 2010, Edinburgh, UK, July 11-14, 2010. Proceedings}, volume
  6175 of {\em Lecture Notes in Computer Science}, pages 113--127. Springer,
  2010.
\newblock URL: \url{http://dx.doi.org/10.1007/978-3-642-14186-7_11}, \href
  {https://doi.org/10.1007/978-3-642-14186-7_11}
  {\path{doi:10.1007/978-3-642-14186-7_11}}.

\bibitem{DBLP:conf/alenex/KutzS07}
Martin Kutz and Pascal Schweitzer.
\newblock Screwbox: a randomized certifying graph-non-isomorphism algorithm.
\newblock In {\em Proceedings of the Ninth Workshop on Algorithm Engineering
  and Experiments, {ALENEX} 2007, New Orleans, Louisiana, USA, January 6,
  2007}. {SIAM}, 2007.
\newblock \href {https://doi.org/10.1137/1.9781611972870.14}
  {\path{doi:10.1137/1.9781611972870.14}}.

\bibitem{10.1007/978-1-4614-1927-3_9}
Leo Liberti.
\newblock Symmetry in mathematical programming.
\newblock In Jon Lee and Sven Leyffer, editors, {\em Mixed Integer Nonlinear
  Programming}, pages 263--283, New York, NY, 2012. Springer New York.
\newblock \href {https://doi.org/10.1007/978-1-4614-1927-3_9}
  {\path{doi:10.1007/978-1-4614-1927-3_9}}.

\bibitem{DBLP:journals/ftcgv/LiuHKG10}
Yanxi Liu, Hagit Hel{-}Or, Craig~S. Kaplan, and Luc~Van Gool.
\newblock Computational symmetry in computer vision and computer graphics.
\newblock {\em Foundations and Trends in Computer Graphics and Vision},
  5(1-2):1--195, 2010.
\newblock URL: \url{http://doi.org/10.1561/9781601983657}, \href
  {https://doi.org/10.1561/9781601983657} {\path{doi:10.1561/9781601983657}}.

\bibitem{DBLP:conf/wea/Lopez-PresaCA13}
Jos{\'{e}}~Luis L{\'{o}}pez{-}Presa, Luis~N{\'{u}}{\~{n}}ez Chiroque, and
  Antonio~Fern{\'{a}}ndez Anta.
\newblock Novel techniques for automorphism group computation.
\newblock In {\em Experimental Algorithms, 12th International Symposium, {SEA}
  2013, Rome, Italy, June 5-7, 2013. Proceedings}, volume 7933 of {\em LNCS},
  pages 296--307. Springer, 2013.
\newblock \href {https://doi.org/10.1007/978-3-642-38527-8_27}
  {\path{doi:10.1007/978-3-642-38527-8_27}}.

\bibitem{McKay201494}
Brendan~D. McKay and Adolfo Piperno.
\newblock Practical graph isomorphism, {II}.
\newblock {\em Journal of Symbolic Computation}, 60(0):94--112, 2014.
\newblock URL:
  \url{http://www.sciencedirect.com/science/article/pii/S0747717113001193},
  \href {https://doi.org/10.1016/j.jsc.2013.09.003}
  {\path{doi:10.1016/j.jsc.2013.09.003}}.

\bibitem{DBLP:journals/csur/MillerDC06}
Alice Miller, Alastair~F. Donaldson, and Muffy Calder.
\newblock Symmetry in temporal logic model checking.
\newblock {\em {ACM} Computing Surveys}, 38(3), 2006.
\newblock URL: \url{http://doi.acm.org/10.1145/1132960.1132962}, \href
  {https://doi.org/10.1145/1132960.1132962}
  {\path{doi:10.1145/1132960.1132962}}.

\bibitem{DBLP:conf/stoc/NeuenS18}
Daniel Neuen and Pascal Schweitzer.
\newblock An exponential lower bound for individualization-refinement
  algorithms for graph isomorphism.
\newblock In Ilias Diakonikolas, David Kempe, and Monika Henzinger, editors,
  {\em Proceedings of the 50th Annual {ACM} {SIGACT} Symposium on Theory of
  Computing, {STOC} 2018, Los Angeles, CA, USA, June 25-29, 2018}, pages
  138--150. {ACM}, 2018.
\newblock \href {https://doi.org/10.1145/3188745.3188900}
  {\path{doi:10.1145/3188745.3188900}}.

\bibitem{DBLP:conf/iccS/RuizG08}
Irene~Luque Ruiz and Miguel~{\'{A}}ngel G{\'{o}}mez{-}Nieto.
\newblock A {Java} tool for the management of chemical databases and similarity
  analysis based on molecular graphs isomorphism.
\newblock In {\em Computational Science - {ICCS} 2008, 8th International
  Conference, Krak{\'{o}}w, Poland, June 23-25, 2008, Proceedings, Part {II}},
  volume 5102 of {\em LNCS}, pages 369--378. Springer, 2008.
\newblock \href {https://doi.org/10.1007/978-3-540-69387-1_41}
  {\path{doi:10.1007/978-3-540-69387-1_41}}.

\bibitem{seress_2003}
{\'{A}}kos Seress.
\newblock {\em Permutation Group Algorithms}.
\newblock Cambridge Tracts in Mathematics. Cambridge University Press, 2003.
\newblock \href {https://doi.org/10.1017/CBO9780511546549}
  {\path{doi:10.1017/CBO9780511546549}}.

\bibitem{DBLP:journals/jmlr/ShervashidzeSLMB11}
Nino Shervashidze, Pascal Schweitzer, Erik~Jan van Leeuwen, Kurt Mehlhorn, and
  Karsten~M. Borgwardt.
\newblock {Weisfeiler-Lehman} graph kernels.
\newblock {\em Journal of Machine Learning Research}, 12:2539--2561, 2011.
\newblock URL: \url{https://dl.acm.org/doi/10.5555/1953048.2078187}, \href
  {https://doi.org/10.5555/1953048.2078187}
  {\path{doi:10.5555/1953048.2078187}}.

\bibitem{Tener:2009:ADI:1925258}
Greg~Daniel Tener.
\newblock {\em Attacks on Difficult Instances of Graph Isomorphism: Sequential
  and Parallel Algorithms}.
\newblock PhD thesis, University of Central Florida, Orlando, FL, USA, 2009.
\newblock AAI3401107.

\end{thebibliography}
\bibliographystyle{plainurl}
\end{document}